\documentclass[12pt,authoryear,review]{elsarticle}

\makeatletter
\def\ps@pprintTitle{%
 \let\@oddhead\@empty
 \let\@evenhead\@empty
 \def\@oddfoot{}%
 \let\@evenfoot\@oddfoot}
\makeatother

\usepackage{amsthm,amssymb,amsmath,mathtools,amsbsy}

\usepackage[dvipsnames]{xcolor}
\usepackage{graphicx,rotate}
\usepackage{tikz}
\usepackage{subcaption}

\usepackage{float}
\usepackage{enumerate}
\usepackage{dsfont}

\usepackage{accents}
\usepackage{pdfcomment}

\usepackage{bigstrut}
\usepackage{multirow}
\usepackage{booktabs}

\usepackage{tikz}
\usepackage{wasysym}

\numberwithin{equation}{section}

\newtheorem{theorem}{Theorem}[section]

\newtheorem{proposition}[theorem]{Proposition}

\theoremstyle{definition}

\newcommand{\expe}[2]{\underset{#1}{\E}\left[{#2}\right]}
\newcommand{\E}{\mathbb{E}}
\newcommand{\1}[1]{\mathds{1}_{\left\{ {#1} \right\} }}
\newcommand{\bm}{\pmb}

\def\F{{\mathcal F}}
\renewcommand{\P}{\mathbb{P}}
\newcommand{\Ynd}[1]{ {Y_{#1}^d \in dy_{#1}^d } }
\newcommand{\Y}{\mathcal{Y}}
\newcommand{\Znd}[2]{{Z_{#1}^d = \theta_{#2} }}

\newcommand{\T}{\intercal}

\newcommand{\bchi}{{\boldsymbol \chi}}
\renewcommand{\bell}{{\boldsymbol \ell}}
\newcommand{\bN}{{\boldsymbol N}}
\newcommand{\bZ}{{\boldsymbol Z}}
\newcommand{\bpi}{{\boldsymbol \pi}}
\newcommand{\blambda}{{\boldsymbol \lambda}}
\newcommand{\bLambda}{{\boldsymbol \Lambda}}
\newcommand{\wlambda}{{\widehat{\lambda}}}

\newcommand{\hlambda}{{\widehat \lambda}}
\newcommand{\hA}{{\widehat A}}

\newcommand{\mcA}{{\mathcal A}}
\newcommand{\mcF}{{\mathcal F}}
\newcommand{\mcG}{{\mathcal G}}

\newcommand{\mfN}{{\mathfrak{N}}}
\newcommand{\mfJ}{{\mathfrak{J}}}

\newcommand{\PP}{{\mathbb P}}

\newcommand{\bUpsilon}{{\boldsymbol \Upsilon}}
\newcommand{\mR}{{\mathds R}}
\newcommand{\mZ}{{\mathds Z}}

\newcommand{\Real}{\mathds{R}}
\newcommand{\eyeM}{\boldsymbol{I}_{(J\times J)}}
\newcommand{\kp}{\kappa^\star}
\newcommand{\bCp}{\bm{C}^\star}

\newcommand{\Dt}{{\Delta t\,}}
\newcommand{\mcB}{{\mathcal B}}
\newcommand{\bmu}{{\boldsymbol{\mu}}}

\graphicspath{{}}

\begin{document}

\begin{frontmatter}

\title{\textbf{Trading algorithms with learning in latent alpha models}\tnoteref{t1}
\\[1em]
{Mathematical Finance, Forthcoming\tnoteref{t2}}
}
\tnotetext[t1]{The authors would like to thank Damir Kinzebulatov for discussions on earlier parts of this work. SJ would like to acknowledge the support of the Natural Sciences and Engineering Research Council of Canada (NSERC), [funding reference numbers RGPIN-2018-05705 and RGPAS-2018-522715].}
\tnotetext[t2]{first posted online Nov 2016 https://ssrn.com/abstract=2871403}

\author[author1]{Philippe Casgrain}
\ead{p.casgrain@mail.utoronto.ca}

\author[author1]{Sebastian Jaimungal}
\ead{sebastian.jaimungal@utoronto.ca}
\address[author1] {Department of Statistical Sciences, University of Toronto}

\begin{abstract}
Alpha signals for statistical arbitrage strategies are often driven by latent factors. This paper analyses how to optimally trade with latent factors that cause prices to jump and diffuse. Moreover, we account for the effect of the trader's actions on quoted prices and the prices they receive from trading. Under fairly general assumptions, we demonstrate how the trader can learn the posterior distribution over the latent states, and explicitly solve the latent optimal trading problem. We provide a verification theorem, and a methodology for calibrating the model by deriving a variation of the expectation-maximization algorithm. To illustrate the efficacy of the optimal strategy, we demonstrate its performance through simulations and compare it to strategies which ignore learning in the latent factors. We also provide calibration results for a particular model using Intel Corporation stock as an example.
\end{abstract}

\end{frontmatter}

\section{Introduction}

The phrase ``All models are wrong, but some are useful'' (Box, 1978) rings true across all areas in finance, and intraday trading is no exception. If an investor wishes to efficiently trade assets, she must use a strategy that can anticipate the asset's price trajectory while simultaneously being mindful of the flaws in her model, as well as the costs borne from transaction fees and her own impact on prices. With all of the complexities in intraday markets, it is no surprise that strategies differ substantially based on what assumptions are made about asset price dynamics. Trading with an incorrect model can be very costly to an investor, and therefore being able to mitigate model risk is valuable.

The availability of information at very high frequencies can help a trader partially overcome the problem of model selection. The information provided from realized trajectories of the asset price and the incoming flow of orders of other traders, allows her to infer which model best fits the observed data, and in turn she may use it to predict future movements in asset prices. Ideally, the trader should be able to incorporate this information in an on-line manner. In other words, the trader should be continuously updating her model as she observes new information, keeping in mind that the market may switch between a number of regimes over the course of the trading period. Furthermore, the trader would like to have some means of incorporating a-priori knowledge about markets into her trading strategy before beginning to trade.

This paper studies the optimal trading strategy for a single asset when there are latent alpha components to the asset price dynamics, and where the trader uses price information to learn about the latent factor. Prices can diffuse as well as jump. The trader's goal is to optimally trade subject to this model uncertainty, and end the trading horizon with zero inventory. By treating the trader's problem as a continuous time control problem where information is partially obscured, we succeed in obtaining a closed form strategy, up to the computation of an expectation that is specific to the trader's prior assumptions on the model dynamics. The optimal trading strategy we find can be computed with ease for a wide variety of models, and we demonstrate its performance by comparing, in simulation, with approaches that that do not make use of learning.

Early works on partial information include \cite{detemple1986asset}, \cite{detemple1991further}, who study optimal technology investment problems (where the states that drive production are obfuscated by gaussian noise); \cite{gennotte1986optimal}, who studies the optimal portfolio allocation problem when returns are hidden but satisfy an Ornstein-Uhlenbeck process; \cite{dothan1986equilibrium}, who analyzes a production and exchange economy with a single unobservable source of nondiversifiable risk; \cite{karatzas1991note}, who studies  utility maximization under partial observations; \cite{rieder2005portfolio}, \cite{bauerle2007portfolio} and \cite{frey2012portfolio}, who study model uncertainty in the context of portfolio optimization and the optimal allocation of assets; and \cite{Papa2018}, who studies an optimal portfolio allocation problem where the drift of the assets are latent Ito diffusions.

There are a few recent papers on partial information that are related to this study. \cite{ekstrom2016optimal} investigates the optimal timing problem associated with liquidating a single unit of an asset when the asset price is a geometric Brownian motion with random (unobserved) drift. \cite{colaneri2016shall} studies the optimal liquidation problem when the asset midprice is driven by a Poisson random measure with unknown mean-measure. \cite{garleanu2013dynamic} study the optimal trading strategy for maximizing the discounted, and penalized, future expected excess returns in a discrete-time, infinite-time horizon problem. In their model, prices contain an unpredictable martingale component, and an independent stationary (visible) predictable component -- the alpha component. \cite{gueant2016portfolio} study models in which the drift of the asset price process is a latent random variable in an optimal portfolio selection setting, for an investor who seeks to maximize a CARA and CRRA objective function, as well as in the cases of optimal liquidation in an Almgren-Chriss like setting.


 The approach we take differs in several ways from the extant literature, but the two key generalizations are: (i) we account for quite general latent factors which drive the drift and jump components in the asset's midprice; and (ii) we include both temporary and permanent impact that the agent's trading has on the market.

 The structure of the remainder of this paper is as follows. Section~\ref{sec: The Model} outlines our modelling assumptions, as well as providing the optimization problem with partial information that the trader wishes to solve. Section~\ref{sec: Filtering} provides the filter which the trader uses to make proper inference on the underlying model driving the data she is observing. Section~\ref{sec: Reducing the Problem} shows that the original optimization problem presented in Section~\ref{sec: The Model} can be simplified to an optimization problem with complete information using the filter presented in Section~\ref{sec: Filtering}. Section~\ref{sec: Solving the DPP } shows how to solve the reduced optimization problem from Section~\ref{sec: Reducing the Problem} and verifies that the resulting strategy indeed solves the original optimization problem. Lastly, Section~\ref{sec: Numerical Examples} provides some numerical examples by applying the theory to a few specific models, and compares the resulting strategy, using simulations, to an alternative which does not learn from price dynamics.

\section{The Model} \label{sec: The Model}

We work on the filtered probability space $\left( \Omega , \mcG = \{ \mcG_t , 0 \leq t \leq T \} , \mathbb{P} \right)$, where $T>0$, and finite, is some fixed time horizon. The filtration $\mcG$ is the natural one generated by the paths of the un-impacted asset midprice process $F=(F_t)_{ {t\in[0,T]} }$, the counting processes for the number of buy and sell market orders which cause price changes, denoted $N^+=(N_t^+)_{ {t\in[0,T]} }$ and $N^-=(N_t^-)_{ {t\in[0,T]} }$, and a latent process $\Theta=\left( \Theta_t \right)_{ {t\in[0,T]} }$. The exact nature of these processes will be provided in more detail in the remainder of the section.

The trader's optimization problem is to decide on a dynamic trading strategy to buy/sell an asset over the course of a trading horizon to maximize some performance criteria. We assume the trader executes orders continuously at a (controlled) rate denoted by $\nu = (\nu_t)_{{t\in[0,T]}}$. The trader's inventory, given some strategy $\nu$, is denoted  ${Q^\nu = (Q_t^\nu)_{{t\in[0,T]}}}$, with the initial condition $Q^\nu_0=\mfN$. $\mfN$ may be zero, positive (a long position), or negative (a short position) -- and, hence, the inventory at time $t$ can be written as
\begin{equation} \label{eq: Inventory Definition}
  Q_t^\nu = \mfN + \int_0^t \nu_u \; du\;.
\end{equation}
The above can be interpreted as the investor purchasing $\nu_t \, d t$ shares over the period $[t,t+dt)$. A positive (negative) value for $\nu_t$ represents the trader buying (selling) the asset. The rate at which the investor buys or sells the asset affects prices through two mechanisms. Firstly, a temporary price impact, which is effectively a transaction cost that increases with increasing trading rate. Secondly, a permanent impact, which incorporates the fact that when there are excess buy orders, prices move up, and excess sell orders, prices move down.

We further assume that other market participants also have a permanent impact on the asset midprice through their own buy and sell market orders (MOs). To model this, we let  $N^\pm$ be doubly stochastic Poisson processes with respective intensity processes $\lambda^\pm=(\lambda_t^\pm)_{{t\in[0,T]}}$, which count the number of market orders that cause prices to move. In the remainder of the paper, we write  $\bN = (N_t^+,N_t^-)_{ {t\in[0,T]} }$.

\subsection{Asset Midprice Dynamics}

To model the permanent price impact of trades, we define two processes $S = (S_t)_{ {t\in[0,T]} }$ and $F = (F_t)_{ {t\in[0,T]} }$ to represent the asset midprice and the asset midprice without the trader's impact, respectively. As shown by \cite{cartea2016incorporating}, intraday permanent price impact (over short time scales) is well approximated by a linear model.  Hence, we write
\begin{equation} \label{eq: Permanent Impact}
  S_t^\nu = F_t + \beta \int_0^t \nu_u \;du
  \;,
\end{equation}
where $\beta >0$ controls the strength of the trader's impact on the asset midprice. Alternatively, one could write this as a pure jump model \[
S_t^\nu = F_t + \beta\;(\mathcal{M}^{+,\nu}_t-\mathcal{M}^{-,\nu}_t)\;,
\]
where
$({\mathcal{M}}^{+,\nu}_t, {\mathcal{M}}^{-,\nu}_t)_{ {t\in[0,T]} }$
are controlled doubly stochastic Poisson processes with $\PP$--intensities $\gamma^+_t = \nu_t\mathds{1}_{\nu_t>0}$ and $\gamma^-_t = -\nu_t\mathds{1}_{\nu_t<0}$, respectively. The results will be identical to that obtained using the continuous model above.

We assume the investor does not have complete knowledge of the dynamics of the asset midprice, nor the rates of arrival of market orders. This uncertainty is modeled by assuming there is a latent continuous time Markov Chain $\Theta = (\Theta_t)_{ {t\in[0,T]} }$  (with $\Theta_t\in\{ \theta_j \}_{j\in\mfJ}$ and $\mfJ=\{1,2,\dots,J\}$), which modulates the dynamics of state variables, but is not observable by the trader. The latent process $\Theta$ is assumed to have a known generator matrix\footnote{The generator matrix $\bm{C}\in\mathbb{R}^{J\times J}$ of a $J$-state continuous time Markov chain $\Theta$ has non-diagonal entries $\bm{C}_{i,j}\geq 0$ if $i\neq j$ and diagonal entries $\bm{C}_{i,i}= -\sum_{j\neq i} \bm{C}_{i,j}$. $\bm{C}$ is defined so that $\mathbb{P} \left( \Theta_t = \theta_j {\lvert} \Theta_0 = \theta_i \right) = \left( e^{t \, \bm{C}} \right)_{i,j}$, where $\left( e^{t \, \bm{C}} \right)_{i,j}$ is element $(i,j)$ of the matrix exponential of $t \, \bm{C}$.} $\bm C$ and the trader places a prior $\pi_0^j=\mathbb{P}\left[ \Theta_0 = \theta_j \right]$, $j\in\mfJ$, on the initial state of the latent process, all estimated e.g., by the EM algorithm (see Section~\ref{sec: EM Algorithm} for details).

Conditional on a path of $\Theta$, the unaffected midprice $F$ is assumed to satisfy the SDE
\begin{equation*}
  dF_t = A(t,F_t,\bN_t; \Theta_t) \; dt + b\;(dN_t^{+} - dN_t^{-}) + \sigma \;dW_t
  \;, \qquad F_0 = F\;,
\end{equation*}
where $N^\pm$ have $\PP$--intensities
\begin{equation}
  \lambda_t^{\pm} = \sum_{j\in\mfJ} \1{\Theta_t = \theta_j} \lambda_t^{\pm,j}\;,
\end{equation}
and $W = ( W_t )_{ {t\in[0,T]} }$ is a $\mathbb{P}$-Brownian Motion. Moreover, we assume that each of the $\{\lambda^{\pm,j}\}_{j\in\mfJ}$ are $\mcF$--adapted processes, where $\mcF\subseteq \mcG$ is the natural filtration generated by the paths of the processes $F$ (note that $\bN$ can be inferred from this filtration, and strategies are therefore also adapted to the paths of $\bN$). Furthermore, we assume $(\Theta_t,F_t,\blambda_t,\bN_t)_{ {t\in[0,T]} }$ is a $\mcG$--adapted Markov process, where $\blambda = ( \{ \lambda_t^{+,j} , \lambda_t^{-,j} \}_{j\in\mfJ} )_{ {t\in[0,T]} }$. The Markov assumption will, after modifying the problem to deal with partial information, allow a dynamic programming principle (DPP) and result in a dynamic programming equation (DPE). We assume that either (i) $\sigma > 0$ or (ii) $\sigma=0$ and $A:=0$ , to prevent cases where the model is driven by a counting process but also has a continuous drift. In case (ii), the asset price may indeed drift, but the drift will be due to imbalance in intensities so that prices remain on a discrete grid. To compress notation, we define the process $A=\left(A_t\right)_{t\in[0,T]}$ where $A_t :=  A(t,F_t,\bN_t; \Theta_t)$ as well as the processes $A^j = \left( A_t^j \right)_{{t\in[0,T]}}$ where $A_t^j := A(t,F_t,\bN_t; \theta_j)$ for each $j\in\mfJ$.
Finally we make the technical assumption that
\begin{equation} \label{eq: Square Integrable A}
  \mathbb{E}\left[ \int_0^T \left( A_u \right)^2 + \left( \lambda_u^+ \right)^2
  + \left( \lambda_u^-\right)^2
   \;du \right] < \infty
  \;.
\end{equation}

This class of intensity models contains, among many others, deterministic intensities, shot-noise processes, and cross-exciting Hawkes processes with finite-dimensional Markov representations\footnote{To achieve this, we may extend $\blambda_t$ to include the state variables necessary for the model to be Markov.}, all modulated by the latent factor(s). We provide some explicit examples in Section~\ref{sec: Numerical Examples} where we also conduct numerical experiments.

The random variable $\Theta_t$ indexes the $J$ possible models for the asset's drift and the rates at which other market participants' market orders arrive. Because $\Theta$ is (potentially) stochastic, it may change over time, hence, so will the underlying model. Furthermore, because $\Theta$ is invisible to the investor, to make intelligent trading decisions, the investor must infer from observations what is the current (and future) underlying model driving asset prices.

\subsection{Cash Process}

The price at which the trader either buys or sells each unit of the asset will be denoted as $\hat S^\nu = (\hat S^\nu_t)_{ {t\in[0,T]} }$. Because there is limited liquidity at the best bid or ask price (the touch), the investor must ``walk the book'' starting at the bid (ask) and buy (sell) her assets at higher (lower) prices as she increases the size of each of her market orders. For tractability, and as \cite{frei2015optimal} (among others) note, a linear model for this `temporary price impact' fits the data well, and adding in concavity, while empirically more accurate, does not improve the $R^2$ beyond $5\%$. Hence, here we adopt a linear temporary price impact model and write the execution price as
\begin{equation}
  \hat S_t = S_t  + a\;\nu_t\;,
\end{equation}
where $a>0$ controls the asset's liquidity, and hence the impact of trades.

The investor's cash process, i.e., the accumulated funds from trading for some fixed strategy $\nu$, is denoted $X^\nu = (X^\nu_t)_{ {t\in[0,T]} }$, and is given by
\begin{equation}
X_t^\nu = X_0- \int_0^ t \nu_u \; \hat S_u^\nu \; du\;.
\end{equation}

\subsection{Objective Criterion} \label{sec: Defining the Objective}

Over the course of the trading window $t\in[0,T]$, the trader wishes to find a trading strategy $\nu\in\mcA$ which maximizes the objective criterion
\begin{equation} \label{eq: Initial Objective}
  \E \left[
    X_T^\nu + Q_T^\nu \left( S_T^\nu - \alpha\, Q_T^\nu \right)
    - \phi \int_0^T \left( Q_u^\nu \right)^2 du
  \right]
  \;,
\end{equation}
where $\mcA$ is the set of admissible trading strategies, here consisting of the collection of all $\mcF$--predictable processes such that $\mathbb E\left[\int_0^T \nu_u^2 \; du \right]< +\infty$.

The objective criterion~\eqref{eq: Initial Objective} consists of three different parts. The first is $X_T^\nu$, which represents the amount of cash  the trader has accumulated from her trading over the period $[0,T)$. Next is the amount of cash received from liquidating all remaining exposure $Q_T^\nu$ at the end of the trading horizon. The value (per share) of liquidating these shares is penalized by an amount $\alpha \,Q_T^\nu$, where $\alpha \geq 0$. The amount $\alpha \,Q_T^\nu$ represents the liquidity penalty taken by the trader if she chooses to sell or buy an amount of assets $Q_T^\nu$ all at once. We eventually take the limit $\alpha\to\infty$ to ensure that the trader ends with zero inventory. The last term $-\phi \int_0^T (Q_u^\nu)^2 \; du$ represents a running penalty that penalizes the trader for having a non-zero inventory throughout the trading horizon, and allows her to control her exposure. This penalty can also be interpreted as the quadratic variation of the book-value of the traders position (ignoring jumps in the asset price), or can be seen as stemming from model uncertainty as shown in \cite{cartea2014algorithmic}.

Note that we take trading strategies to be $\mcF$--predictable. 
$\mcF$--predictability ensures that the trader does not have access to any information regarding the path of the process $\Theta_t$, which governs the model driving the asset midprice drift and the intensities of $\bN_t$. As well, $\mcF$--predictability prevents the trader for foreseeing a jump occurring at the same instant in time -- in other words her decisions are based on the left limits of $F$, and hence also $\bN$. Because admissible controls are $\mcF$--predictable, and not $\mcG$--predictable (the full filtration), maximizing \eqref{eq: Initial Objective} is a control problem with partial information.

Solving control problems with partial information is very difficult to do directly, because most tools that are used to work with the case of complete information no longer work. The former requires an indirect approach in which, firstly, we find an alternate $\mcF$--adapted representation for the dynamics of the state variable process, and secondly, we extend the state variable process so that it becomes Markov when using Markov controls. The key step in this approach is to find the best guess for $\Theta_t$ conditional on the reduced filtration available at that time.

\section{Filtering} \label{sec: Filtering}

Because the investor cannot observe $\Theta_t$, she wishes to formulate a best guess for its value. The best possible guess for the distribution of $\Theta_t$ will be the distribution of $\Theta_t$ conditional on the information accumulated up until that time. Therefore, she wishes to compute
\begin{equation*}
  \pi_t^j = \E \left[ \left.\, \1{\Theta_t=\theta_j} \, \right | \, \mcF_t \right] ,\quad \forall j \in\mfJ
  \;.
\end{equation*}
The filter process $\bpi = (\{\pi_t^j\}_{j\in\mfJ})_{ {t\in[0,T]} }$ is  $\mcF$--adapted with initial condition $\bpi^0=\{\pi_0\}_{j\in\mfJ}$. It represents the posterior latent state distribution (given all information accumulated by the investor up until $t$).

\begin{theorem} \label{th: SDE Filter Theorem}
  Let us assume that the Novikov condition
  \begin{equation} \label{eq:NovikovCondition}
  \mathbb{E}\left[ \exp\left\{\int_0^T \left( A_u \right)^2 + \left( \lambda_u^+ \right)^2
  + \left( \lambda_u^-\right)^2
   \;du \right\} \right] < \infty
  \;
  \end{equation}
  holds.
  Then the filter $\bpi$ admits a representation with components
  \begin{equation}
    \pi_t^i = \Lambda_t^i \left/ \sum_{j=1}^J \Lambda_t^j \right.
    \;,
  \end{equation}
  where $\bLambda =(\{\Lambda_t^j\}_{j\in\mfJ})_{ {t\in[0,T]} }$. If $\sigma>0$, for each $i\in\mfJ$, $\Lambda_t^i$ solves the SDE
  \begin{equation}\label{eq:Lambda_SDE}
  \begin{split}
\frac{d\Lambda_t^i}{\Lambda_{t-}^i} =&\;
  \sigma^{-2}  A_{t-}^i \left( dF_t - b \,( dN_t^+ - dN_t^-)  \right) \\ \
  & \quad +
  (\lambda_{t-}^{+,i} -1)(dN_t^+ - dt) +
  (\lambda_{t-}^{-,i} -1)(dN_t^- - dt)
  + \sum_{j\in\mfJ} \left( \frac{\Lambda_{t-}^j}{\Lambda_{t-}^i} \right) C_{i,j} \;dt\;
  \end{split}
  \end{equation}
with initial condition $\bLambda_0 = \bpi_0$. If $\sigma = 0$ and $A_t:= 0$, for each $i\in\mfJ$, $\Lambda_t^i$ solves the SDE
  \begin{equation}\label{eq:Lambda_SDE_sigma0}
  \begin{split}
    \frac{d\Lambda_t^i}{\Lambda_{t-}^i} =
  (\lambda_{t-}^{+,i} -1)(dN_t^+ - dt) +
  (\lambda_{t-}^{-,i} -1)(dN_t^- - dt)
  + \sum_{j\in\mfJ} \left( \frac{\Lambda_{t-}^j}{\Lambda_{t-}^i} \right) C_{i,j} \;dt\;,
  \end{split}
  \end{equation}
  with the same initial condition.
  \begin{proof}
    See \ref{sec: Proof of SDE Filter Theorem}.
  \end{proof}
\end{theorem}

The process $\bLambda$ admits a simple closed form solution when $\bm C = \boldsymbol 0$. This case corresponds to when the latent regimes are constant over the trading period $[0,T]$ -- in others words, the case of parameter uncertainty, but the model does not switch between regimes throughout the trading horizon. When $\bm C \ne \boldsymbol 0$, solutions to the filter can be approximated reasonably well for most purposes by using methods outlined in~\cite{DiscreteWonhamApproximation}, which will be discussed further in Section~\ref{sec: Numerical Examples}.

An SDE also exists for the normalized version of the filter $\bpi_t$, however, for simplicity, we keep track of the processes $\bLambda$, and define the function (with a slight abuse of notation) $\pi^j: \mR^J_+ \mapsto [0,1]$ via
\begin{equation}
  \pi^j(\bUpsilon) = \Upsilon^j \left/\sum_{i=1}^J \Upsilon^i\right.\;,\qquad \forall \bUpsilon\in\mR^J_+,
\end{equation}
so that $\pi_t^j = \pi^j(\bLambda_t)$.
This choice of mapping $\bLambda$ into $\bpi$ guarantees that $\sum_{j=1}^J \pi_t^j=1$, even when numerically approximating \eqref{eq:Lambda_SDE}.

\section{\texorpdfstring{$\F$}--Dynamics Projection} \label{sec: Reducing the Problem}

In this section, we show there exists an $\mcF$--adapted representation for the price dynamics, and the intensity processes.
The sequence of arguments resemble those found in~\cite[Section 3]{bauerle2007portfolio}, adapted to the case where the observable process contains both jump and diffusive terms.

First, define the $\mcG$--adapted martingales $\bm M = (M_t^+,M_t^-)_{ {t\in[0,T]} }$ to be the compensated versions of the Poisson processes $\bN$, i.e.,
\begin{equation}
  M_t^\pm = N_t^\pm - \int_0^t \lambda_u^\pm\, du
  \;.
\end{equation}
The theorem below provides the necessary ingredients to provide the $\mcF$--adapted representations of the state processes.
\begin{theorem}\label{th: Predictable Representation Theorem }
If $\sigma>0$, define the processes $\widehat W=(\widehat W_t)_{ {t\in[0,T]} }$, $\widehat{\boldsymbol{M}}=({\widehat{M}}^+_t, {\widehat{M}}^-_t)_{ {t\in[0,T]} }$ by the following relations
\begin{subequations}
\begin{align}
    \widehat W_t &= W_t + \sigma^{-1} \int_0^t \left( A_u - \hA_u \right) \; du
    \;,\\
    \widehat M_t^\pm &= M_t^\pm + \int_0^t \left( \lambda_u^\pm - \hlambda^\pm_t \right)\; du \label{eqn:defMhat}
    \;,
\end{align}
\end{subequations}
where $\hA=(\hA_t)_{t\in[0,T]}$ and $\hlambda^\pm=(\hlambda^\pm_t)_{t\in[0,T]} $ are the filtered drift and intensities, defined as $\hA_t := \E\left[ A_t \lvert \mcF_t \right]$ and $\hlambda^\pm_t := \E\left[ \lambda^\pm_t \lvert \mcF_t \right]$.
Then,
  \begin{enumerate}[(A)]
    \item the process $\widehat W$ is an $\mcF$--adapted $\mathbb{P}$--Brownian motion;
    \item the process $\widehat{\boldsymbol{M}}$ is an $\mcF$--adapted $\mathbb{P}$--martingale; and
    \item $[\widehat W,\widehat M^\pm]_t = 0$ and $[\widehat M^+,\widehat M^-]_t = 0$, $\mathbb{P}$--almost surely.
    \item $N^\pm$ are $\mathcal{F}$--adapted doubly stochastic Poisson processes with $\mathbb{P}$-intensities $\hlambda^\pm$.
  \end{enumerate}

  If $\sigma = 0$ and $A:=0$, define $\widehat{\boldsymbol{M}}=({\widehat{M}}^+_t, {\widehat{M}}^-_t)_{ {t\in[0,T]} }$ as in \eqref{eqn:defMhat}. Then, (B) and (D) hold and $[\widehat M^+,\widehat M^-]_t = 0$, $\mathbb{P}$--almost surely.

  \begin{proof}
    See \ref{sec: Proof of Predictable Representation Theorem}.
  \end{proof}
\end{theorem}

Theorem~\ref{th: Predictable Representation Theorem } tells us that $N^\pm$, in addition to being viewed as a $\mcG$--adapted doubly stochastic Poisson process with $\PP$--intensity of $\lambda^\pm$, can be viewed as an $\mcF$--adapted doubly stochastic process with $\PP$-intensity $\wlambda^\pm$. That is, $N^\pm$ is a doubly stochastic Poisson process with respect to both the $\mcF$ and $\mcG$ filtrations, but with differing intensities.

Theorem~\ref{th: Predictable Representation Theorem } allow us to represent the dynamics of $F$ in their $\mcF$--predictable form as
\begin{equation} \label{eq: F predictable dynamics}
  dF_t = \left( \hA_t + b \, ( \widehat \lambda_t^{+} - \widehat \lambda_t^{-} ) \right) dt + b \left(  d{\widehat{M}}_t^+ - d{\widehat{M}}_t^- \right) + \sigma \;d\widehat W_t
  \;.
\end{equation}
Let us also note that because $A_t = \sum_{j\in\mfJ} \1{\Theta_t = \theta_j} \, A_t^j$ and $\lambda_t^\pm = \sum_{j\in\mfJ} \1{\Theta_t = \theta_j} \, \lambda_t^{j,\pm}$, because $\{A^j:j\in\mfJ\}$ are $\mcF$--adapted, we may take a conditional expectation with respect to $\mcF_t$ to yield that $\hA_t = \sum_{j\in\mfJ} \pi_t^j A_t^j$ and $\hlambda_t^\pm = \sum_{j\in\mfJ} \pi_t^j \, \lambda_t^{j,\pm}$. Therefore we may define the functions, $\hA:\mR_+\times \mR\times\mZ_+^2\times \mR^J_+\mapsto\mR$ and $\wlambda^\pm:\mR_+^{2J}\times \mR^J_+\mapsto\mR_+$ as
\begin{equation}
\label{eq:Def-hA-hlambda}
\hA(t,F,\bN,\bLambda) := \sum_{j\in\mfJ} \pi^j(\bLambda) \; A(t,F,\bN,\theta_j) \quad
 \text{and} \quad
 \hlambda^\pm(\bm\lambda,\bm\Lambda) := \sum_{j\in\mfJ} \pi^j(\bLambda) \; \lambda^{\pm,j}\;,
\end{equation}
so that $\hA_t = \hA(t,F_t,\bN_t,\bLambda_t)$ and $\hlambda_t^\pm = \hlambda^\pm(\bm\lambda_t,\bm\Lambda_t)$.

Hence, the collection of processes $(F,\bN,\blambda,\bLambda)$ are $\mcF$--adapted. The optimal control problem corresponding to maximizing~\eqref{eq: Initial Objective}, within the admissible set, can therefore be regarded as a problem with complete information with respect to the extended state variable process $(S^\nu,F,\bN,X^\nu,Q^\nu,\blambda,\bLambda)$. The joint dynamics of this state process are all $\mcF$--adapted and do not depend on the process $\Theta$. Therefore, the dynamics of the extended state process are completely visible to the investor, which reduces the control problem with partial information, in which we did not know the dynamics of the state variables, into a control problem with full information.

In the next section, we solve this control problem by using the fact that the extended state variable dynamics are $\mcF$--adapted for each $\nu\in\mcA$. Hence, the dynamic programming principle can be applied to the optimization problem~\eqref{eq: Initial Objective} and we derive a dynamic programming equation for the new problem.

\section{Solving the Dynamic Programming Problem} \label{sec: Solving the DPP }

\subsection{The Dynamic Programming Equation}

Using the definitions for $S_t^\nu$ and $Q_t^\nu$ in \eqref{eq: Permanent Impact} and~\eqref{eq: Inventory Definition}, we can write $S_t^\nu$ as
\begin{equation}
  S_t^\nu = F_t + \beta \left( Q_t^\nu - \mfN \right)
  \;,
\end{equation}
as well we can write
\begin{equation} \label{eq: alternative X definition}
  dX_t^\nu = -\,\nu_t \left( F_t + \beta \left( Q_t^\nu - \mfN \right) - a\,\nu_t \right) \, dt
  \;,
\end{equation}
which allows $X^\nu$ to be defined independently of $S^\nu$. Hence, the trader's objective criterion~\eqref{eq: Initial Objective} becomes
\begin{equation} \label{eq: Initial Objective alternative definition}
  \E \left[
    X_T^\nu + Q_T^\nu \left( F_T + \beta \left( Q_T^\nu - \mfN \right) - \alpha\, Q_T^\nu \right)
    - \phi \int_0^T \left( Q_u^\nu \right)^2 du
  \right]
  \;.
\end{equation}
With $X$  given by \eqref{eq: alternative X definition}, the trader's objective function does not depend on the value of the process $S^\nu$. For the remainder of this section, we will use the above definition for the trader's objective criterion.

To optimize the objective criterion~\ref{eq: Initial Objective alternative definition}, we use the fact that $\forall \;\nu\in\mcA$, the $(3J+5)$-dimensional state variable process ${\boldsymbol {Z}}^\nu = (F,\bN,X^\nu,Q^\nu,\blambda,\bLambda)$ is $\mcF$--adapted and, hence, has dynamics visible to the trader. First, let us define the functional
\begin{equation} \label{eq: Value Function Definition}
  H^\nu(t,\bZ)
  = \underset{t,{\bZ}}{\mathbb{E}} \left[ X_T^\nu + Q_T^\nu \left( F_T + \beta \left( Q_T^\nu - \mfN \right)  - \alpha Q_T^\nu \right)
    - \phi \int_t^T \left( Q_u^\nu \right)^2 du \right]\,,
\end{equation}
and the value function
\begin{equation} \label{eq: Optim Value Function Definition}
  H(t,\bZ) = \sup_{\nu\in\mcA} H^\nu(t,\bZ)
  \;,
\end{equation}
where we use $\underset{t,{\bZ}}{\mathbb{E}}[\;\bm\cdot\;]$ to represent the expected value given the initial conditions ${\mathcal{\boldsymbol Z}}_{t^-}^\nu = {\bZ}=(F,\bN,X,Q,\blambda,\bLambda)\in \mathcal D$, where $\mathcal D=\mathds R \times \mathds Z_+^2 \times \mR\times \mR \times \mR_+^{2J} \times R_+^{J}$. The definition of $H^\nu$ implies that $H^\nu(0,\bZ_0)$, where $\bZ_0=(F,\bm 0 , X , \mfN , \blambda, \bm \pi_0)$, is the objective criterion defined in equation~\eqref{eq: Initial Objective alternative definition}. Furthermore, a control $\nu^\star\in\mcA$ is optimal and solves the optimization problem described in Section~\ref{sec: Defining the Objective} if it satisfies
\begin{equation}
  H^{\nu^{\star}}(0,\bZ_0) = H(0,\bZ_0)
  \;.
\end{equation}

Given the $\mcF$--adapted version of the dynamics of the state variables, for any Markov admissible control $\nu\in \mcA$, there exists some function $g:\mR_+ \times \mathcal D$, such that $\nu_t = g(t,\bZ_t^\nu)$. For such controls, the function $H$ must satisfy the Dynamic Programming Principle and the Dynamic Programming Equation (DPE) (see, e.g., \cite[Chapter 3]{PhamControlBook}) applies. The DPE for our specific problem suggests that $H$ satisfies the PDE
\begin{equation} \label{eq: unsuped HJB}
\left\{
  \begin{aligned}
    -\phi\, q^2 + \sup_{\nu\in\mathds{R}} \left\{
      (\partial_t + \mathcal{L}^{\nu})\,H(t,\bZ)
    \right\} &=0\;,\\
    H(T,\bZ) &= X + Q\,( F + \beta \left( Q - \mfN \right) - \alpha \,Q)\;,
  \end{aligned}
\right.
\end{equation}
where $\mathcal{L}^\nu$ is the infinitesimal generator for the state process ${\mathcal{\boldsymbol Z}}^\nu$ using the predictable representation for the dynamics of $F$ and the intensity of $\bN$, given a fixed control $\nu$. Furthermore, the operator $\mathcal{L}^\nu$ acts on functions $f:\mR_+ \times \mathcal D\mapsto\mR$, once differentiable in $t$, twice differentiable in $F,\blambda,\bLambda$ and all (componentwise) cross-derivatives, and once differentiable in $X,Q$, as follows
\begin{align*} \label{eq: L Bar Differential Operator}
  \mathcal{L}^{\nu} f &= \nu\,\partial_Q f - \nu\,( F + \beta \left( Q - \mfN \right) + a \,\nu )\,\partial_X f + \mathcal{\bar L}f \;,
\end{align*}
where $\mathcal{\bar L}$ is the infinitesimal generator of the process $(F,\bN,\blambda,\bLambda)$ using its $\F$--predictable representation, which is independent of the control $\nu$. This portion of the generator can be fairly generic because we have not specified the precise nature of the dynamics of the intensity processes -- which is the impetus for separating this portion of the generator.

\subsection{Dimensional Reduction}

The Dynamic Programming Equation~\eqref{eq: unsuped HJB} can be simplified by introducing the ansatz
\begin{equation*}
  H(t,\bZ) = X + Q\,\left( F + \beta (Q - \mathfrak N ) \right) + h(t, \bell(\bZ)) \;,
\end{equation*}
where for $\bZ=(F,\bN,X,Q,\blambda,\bLambda)\in\mathcal D$, we write $\bell(\bZ) = (F,\bN,
Q,\blambda,\bLambda)\in\mR \times \mathds Z_+^2 \times \mR \times \mR^{2J}_+\times \mR^J$. The
PDE~\eqref{eq: unsuped HJB} then simplifies significantly to a PDE for $h$,
\begin{equation}
\label{eq: Reduced Ansatz 1 HJB PDE}
\left\{
  \begin{array}{rl}
    0 =&\!\! -\phi \, Q^2 + \left( \partial_t + \mathcal{\bar L} \right) h(t,\bell)
    + Q\,\left( \hA(t,F,\bN,\bLambda)
    + b \left( \widehat \lambda^+(\blambda,\bLambda ) -
    \widehat \lambda^-(\blambda,\bLambda ) \right)  \right)
    \\&\!\!+ \displaystyle\sup_{\nu \in \mathds{R}} \left\{ (\beta\,Q + \partial_Q h)\,\nu - a\, \nu^2 \right\} \\
    h(T,\bell) =&\!\! -\alpha\, Q^2
    \;,
  \end{array}
  \right.
\end{equation}
where the functions $\hA$ and $\hlambda^\pm$ are defined in equation~\eqref{eq:Def-hA-hlambda}. This PDE implies that the feedback control for this problem should be
\begin{equation} \label{eq: Feedback Control}
\nu^\star(t,\bZ) = \tfrac{1}{2a} \left(\; \beta \,Q + \partial_Q h(t,\bell(\bZ))\; \right){}
  \;.
\end{equation}
In other words, the second line of the PDE~\eqref{eq: Reduced Ansatz 1 HJB PDE} attains its supremum at $\nu^\star$ defined above.

\subsection{Solving the DPE}
The ansatz provided above permits us to indeed find a solution to the PDE~\eqref{eq: unsuped HJB} which is presented in the proposition that follows.
\begin{proposition}[Candidate Solution]\label{prop: Candidate Solution}
  The PDE~\eqref{eq: unsuped HJB}, admits the classical solution $H$
  \begin{equation*}
  H(t,\bZ)
  =
  X + Q\,\left( F + \beta\,(Q-\mfN) \right) +
  h_0(t,\bchi(\bZ)) + Q\: h_1(t, \bchi(\bZ))
  + Q^2\: h_2(t) \;,
  \end{equation*}
  where $\bchi(\bZ) = (F,\bN, \blambda, \bLambda)$. Let, $\underset{t,\bchi}{\E} [\;\bm\cdot\;]$ denote expectation conditional on the initial conditions $(F_{t^-},\bN_{t^-},\blambda_{t^-},\bLambda_{t^-})=\bchi$, and define the constants
$\gamma = \sqrt{\phi/a\,}$ and
$\zeta = \tfrac{\alpha - \tfrac{1}{2}\beta  + a \gamma}{\alpha - \tfrac{1}{2} \beta - a \gamma}$. We have that
\\
(i) if $\alpha - \tfrac{1}{2}\beta \neq \sqrt{a\phi}$, then
  \begin{subequations}
  \begin{align}
    h_2(t) &= -a\,\gamma \left( \frac{\zeta e^{\gamma\left( T-t \right)}  + e^{-\gamma\left( T-t \right)}}{\zeta e^{\gamma\left( T-t \right)} - e^{-\gamma\left( T-t \right)}} \right) + \tfrac{1}{2} \beta \\
    h_1(t,\bchi) &=
    \int_t^T
    \underset{t,\bchi}{\mathbb{E}}
    \left[
    \hA_u + b\,(\wlambda_u^{+} - \wlambda_u^{-})
    \right]
    \, \left( \frac{\zeta e^{\gamma\left( T-u \right)}  - e^{-\gamma\left( T-u \right)}}{\zeta e^{\gamma\left( T-t \right)} - e^{-\gamma\left( T-t \right)}} \right) du \label{eq:h1_solution} \\
    h_0(t,\bchi) &= \tfrac{1}{4a} \;
    \underset{t,\bchi}{\E} \left[ \int_t^T \left( h_1(u,\bchi_u) \right)^2 \; du \right].
  \end{align}
\end{subequations}
(ii) if $\alpha - \tfrac{1}{2}\beta = \sqrt{a\phi}$, then
\begin{subequations}
  \begin{align}
    h_2(t) &= -a\,\gamma + \tfrac{1}{2} \beta \\
    h_1(t,\bchi) &=
    \int_t^T
    \underset{t,\bchi}{\mathbb{E}}
    \left[
    \hA_u + b\,(\wlambda_u^{+} - \wlambda_u^{-})
    \right]
    \, e^{-\gamma (u-t)} du \label{eq:h1_solution_2} \\
    h_0(t,\bchi) &= \tfrac{1}{4a}\;
    \underset{t,\bchi}{\E} \left[ \int_t^T \left( h_1(u,\bchi_u) \right)^2 \; du \right]\,,%
  \end{align}%
\end{subequations}%
where $\bchi_t=\bchi(\bZ_t)$.
\end{proposition}
\begin{proof}
See \ref{proof: Prop Candidate}.
\end{proof}

For the remainder of the paper, we will concern ourselves with the case where $\alpha - \tfrac{1}{2}\beta > \sqrt{a\phi}$, because in most applications the trader wishes to completely liquidate by the end of the trading horizon, and so $\alpha\gg1$, while $\sqrt{a\phi}$ is comparatively small.

The above proposition and equation~\eqref{eq: Feedback Control} suggest that the optimal trading speed the investor should employ is
\begin{align} \label{eq: Candidate optimal cnotrol}
  \nu^\star_t
  = \tfrac{1}{2a}\, \left( \, 2\,h_{2}(t) + \beta \, \right) \,Q_t^{\nu^\star} +
  \tfrac{1}{2a}\;h_{1}(t,\bchi_t)
  \;.
\end{align}
This optimal trading strategy is a combination of two terms (i) the classical Almgren-Chriss (AC) liquidation strategy represented by $\tfrac{1}{2a}  \left( \; 2h_{2}(t) + \beta \; \right)\,Q_t^{\nu^\star}$; and (ii) a term which adjusts the strategy based on expected future midprice movements, represented by $\frac{1}{2a}\,h_1(t,\bchi_t)$. From the representation of $h_1$ in \eqref{eq:h1_solution} (or \eqref{eq:h1_solution_2}), this latter term is the weighted average of the expected future drift of the asset's midprice. Therefore if, based on her current information, the trader believes that the asset midprice drift will remain largely positive for the remainder of the trading period, she will buy more of the asset relative to the AC strategy. This is reasonable, because she knows she will be able to sell the asset at a higher price once asset prices have risen. The exact opposite occurs when she expects the asset price drift to remain mostly negative over the rest of the trading period.

The result in \eqref{eq: Candidate optimal cnotrol} illustrates how the investor uses the filter $\bpi_t$ for the posterior probability of what latent state is currently prevailing, to consistently update her strategy based on her predictions of the future path of the asset midprice. Moreover, the solution here closely resembles the result obtained by~\cite{cartea2016incorporating}, however, it explicitly incorporates latent information and jumps in the asset price.

Computing the expectation appearing in $h_1$ directly is not easy. There is, however, an alternate representation of this expectation. For any $u \geq t$, we have
\begin{equation}
\underset{t,\bchi}{\mathbb{E}}
    \left[
    \hA_u + b(\wlambda_u^{+} - \wlambda_u^{-})
    \right]
 =
    \sum_{j\in\mfJ} \pi^j(\bLambda) \underset{t,\bchi,\theta_j}{\mathbb{E}}
    \left[
    A_u + b(\lambda_u^{+} - \lambda_u^{-})
    \right]
    \;,
\end{equation}
where $\underset{t,\bchi,\theta_j}{\mathbb{E}}[\;\bm\cdot\;]$ denotes expectation conditioning on the initial condition $(F_{t^-},\bN_{t^-},\blambda_{t^-},\bLambda_{t^-})=\bchi$ and $\Theta_t = \theta_j$. The alternative form in the rhs above is almost always easier to compute than a direct computation of the lhs.

Next, we provide a verification theorem showing that the candidate solution in Proposition~\ref{prop: Candidate Solution} is exactly equal to the value function $H$ defined in equation~\eqref{eq: Optim Value Function Definition}.
\begin{theorem}[Verification Theorem] \label{theorem: Verification Theorem}
    Suppose that $h$ is the solution to the PDE~\eqref{eq: Reduced Ansatz 1 HJB PDE}, and that $\alpha-\frac{1}{2}\beta \neq a \gamma$. Let $\widehat H(t,\bZ) = X + Q\,\left( F + \beta (Q - \mfN) \right) + h(t,\bell(\bZ))$, where $\bell(\bZ) = (F,\bN,Q,\blambda,\bLambda)$.\\
	Then $\widehat H$ is equal to the value function $H$ defined in~\eqref{eq: Optim Value Function Definition}. Furthermore the control
\begin{align}
  \nu_t^{\star}
  = \tfrac{1}{2a} \left( \; 2\,h_{2}(t) + \beta \; \right)  \, Q_t^{\nu^\star} +
  \tfrac{1}{2a}\;h_{1}(t,\bell(\bZ_t))
\end{align}
is optimal and satisfies
\begin{equation}
  H(t,\bZ) = H^{\nu^\star}(t,\bZ)
  \;.
\end{equation}
\end{theorem}
\begin{proof}
  See \ref{sec: Proof Verification Theorem}.
\end{proof}

The theorem above guarantees that the control provided above indeed solves the optimization problem presented in Section~\ref{sec: Defining the Objective}. In retrospect, the optimal control to the trader's optimization problem with partial information is a Markov control. The key steps were to introduce the predictable representation for the dynamics of the process $F$, and to extend the original state process to include the unnormalized posterior distribution $\bLambda$ of the latent states $\Theta$.

\subsection{Zero Terminal Inventory}

A useful limiting case is when the trader is forced to eliminate her market exposure before time $T$. This corresponds to taking the limit $\alpha \rightarrow \infty$ and the resulting optimal control simplifies to
\begin{equation} \label{eq:Optimal Control Alpha Limit}
\begin{split}
    \lim_{\alpha \rightarrow \infty } \nu_t^\star =
  & - \gamma \,\coth \left( \gamma (T-t) \right)\;Q_t^{\nu^\star}
  \\
  &+ \frac{1}{2a}\sum_{j\in\mfJ} \pi^j(\bLambda_t) \int_t^T \underset{t,\bchi_t, \theta_j}{\mathbb{E}}
    \left[
    A_u + b(\lambda_u^{+} - \lambda_u^{-})
    \right]
    \left(\frac{\sinh\left(\gamma(T-u)\right)}{\sinh\left(\gamma(T-t)\right)}\right) du
  \;.
\end{split}
\end{equation}
A second interesting case is to additionally take the limit of no running inventory penalty, in which case the optimal strategy results in
\begin{equation}
\begin{split}
    \lim_{\phi \rightarrow 0} \lim_{\alpha \rightarrow \infty } \nu_t^\star =  &- \frac{1}{T-t} \; Q_t^{\nu^\star}
    \\
  &+\frac{1}{2a}\sum_{j=1}^J \pi^j(\bLambda_t) \int_t^T \underset{t,\bchi_t, \theta_j}{\mathbb{E}}
    \left[
    A_u + b(\lambda_u^{+} - \lambda_u^{-})
    \right]
    \left(\frac{T-u}{T-t}\right) du\,.
\end{split}
\end{equation}
This strategy corresponds to a time weighted average price (TWAP) strategy plus an adjustment for the weighted expected future drift of the asset's midprice.

All of the expressions above for the optimal control can be computed in closed form for a large variety of models. In the next section we provide two explicit, and useful, examples together with numerical experiments to illustrate the strategies dynamic behaviour.

\section{Numerical Examples} \label{sec: Numerical Examples}

In this section we will carry out some numerical experiments to test the performance of the optimal trading algorithms developed in Section~\ref{sec: Solving the DPP }. The examples show how the optimal trading performs using situations for two model set-ups.

\subsection{Mean-Reverting Diffusion}

This section investigates the case where the trader wishes to liquidate her inventory before some specified time $T$. The asset price is assumed to be a pure diffusive Ornstein-Uhlenbeck process -- alternatively, one can think of this midprice as the number of long-short position in a pairs trading strategy. The trader knows the volatility and rate of mean reversion, but does not know the level at which prices revert to. In this example, the mean-reversion level will remain constant over the course of the trading period $[0,T]$. More specifically, we assume that the asset midprice in USD has the dynamics
\begin{equation}\label{eq:OU-model}
  dF_t = \kappa\, ( \Theta - F_t ) \; dt + \sigma\, dW_t
  \;,
\end{equation}
where $\Theta$ is a random variable taking values in the set $\{ \theta_j \}_{j\in\mfJ}$ with probabilities $\{\pi_0^j\}_{j\in\mfJ}$. It remains constant over time but its value is hidden from the trader. This model does not contain any jumps so we can ignore the variables $\bN$ and $\blambda$.

As mentioned in Section~\ref{sec: Filtering}, there exists an exact closed form for the filter when $\Theta_t$ is constant in time. For the regime switching OU model in \eqref{eq:OU-model}, the exact solution for the un-normalized filter is
\begin{equation}
  \Lambda_t^j =
  \pi_0^j \;\exp\left\{\sigma^{-2} \left( \textstyle\int_0^t \kappa (\theta_j - F_u)\; dF_u - \tfrac{1}{2} v\int_0^t \kappa^2 ( \theta_j - F_u )^2 \;du \right)\right\}
  \,, \;\;\forall j\in\mfJ.
\end{equation}
Because, in practice, $F$ is observed only discretely, the integrals above are approximated using the appropriate Riemann sums. The more frequently the trader observes $F$, the more accurate the filter will be.

The solution to the optimal control when $\alpha \rightarrow \infty$ can be computed exactly as
\begin{equation*}
\begin{split}
  \nu^\star (t,F,Q,\bLambda)= &  -\gamma \coth \left( \gamma \, (T-t) \right)\,Q  \\
  &+\sum_{j=1}^J
  \pi^j(\bLambda)\;\kappa\,(F-\theta_j)
  \int_t^T
     e^{-\kappa(u-t)}
     \frac{\sinh\left(\gamma(T-u)\right) }{\sinh\left(\gamma(T-t)\right)} du
  \;.
\end{split}
\end{equation*}

For the simulations, here, we assume there are two possible values the asset price mean-reverts to, so that $J=2$ and we set $\theta_1 = \$ 4.85$ and $\theta_2= \$ 5.15$. Furthermore, we assume the investor has an equal prior on the two possibilities, so that $\pi_0^1 = \pi_0^2 = 0.5$. The remaining parameters used in the simulation are provided in Table~\ref{tbl:OU-params}.
\begin{table}[H]
  \centering
  \begin{tabular}{llll}
  $\mfN=10^4$, & $F_0 = \$ 5$, & $\sigma=0.15$, &$\beta = \$ 10^{-3}$, \\
  $\kappa = 2$, & $a=\$ 10^{-5}$, & $\phi=2\times10^{-5}$.
\end{tabular}
\caption{The parameters in the OU model. All of the time-sensitive parameters are defined on an hourly scale. \label{tbl:OU-params}}
\end{table}
When simulating sample paths, we generate paths using $\Theta=\$ 5.15$. The trader will need to detect this value as she observes the price path.
\begin{figure}[t!]
    \centering
    \begin{subfigure}[t]{0.48\textwidth}
        \centering
        \includegraphics[width=\textwidth]{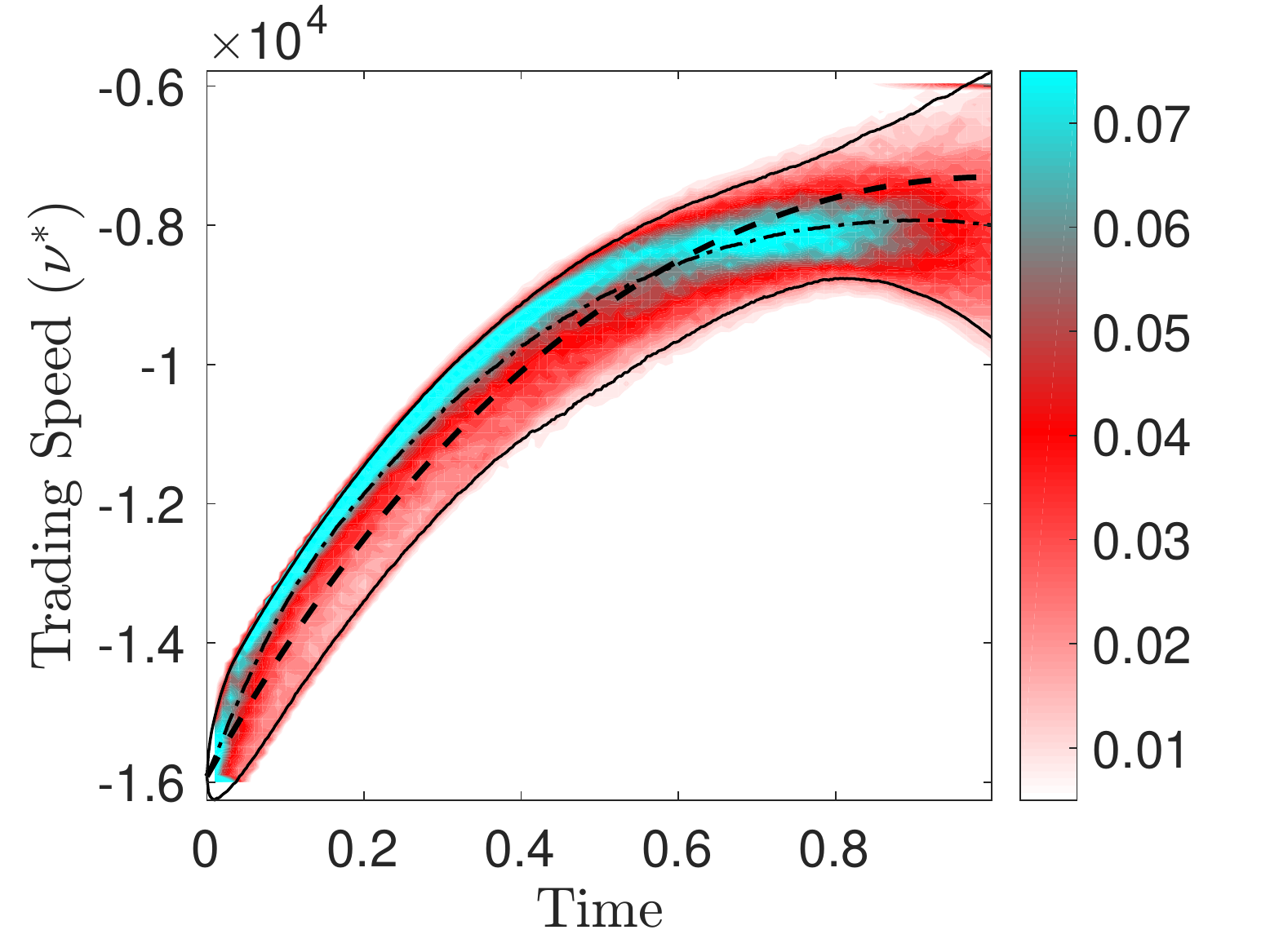}
    \end{subfigure}%
    ~
    \begin{subfigure}[t]{0.48\textwidth}
        \centering
        \includegraphics[width=\textwidth]{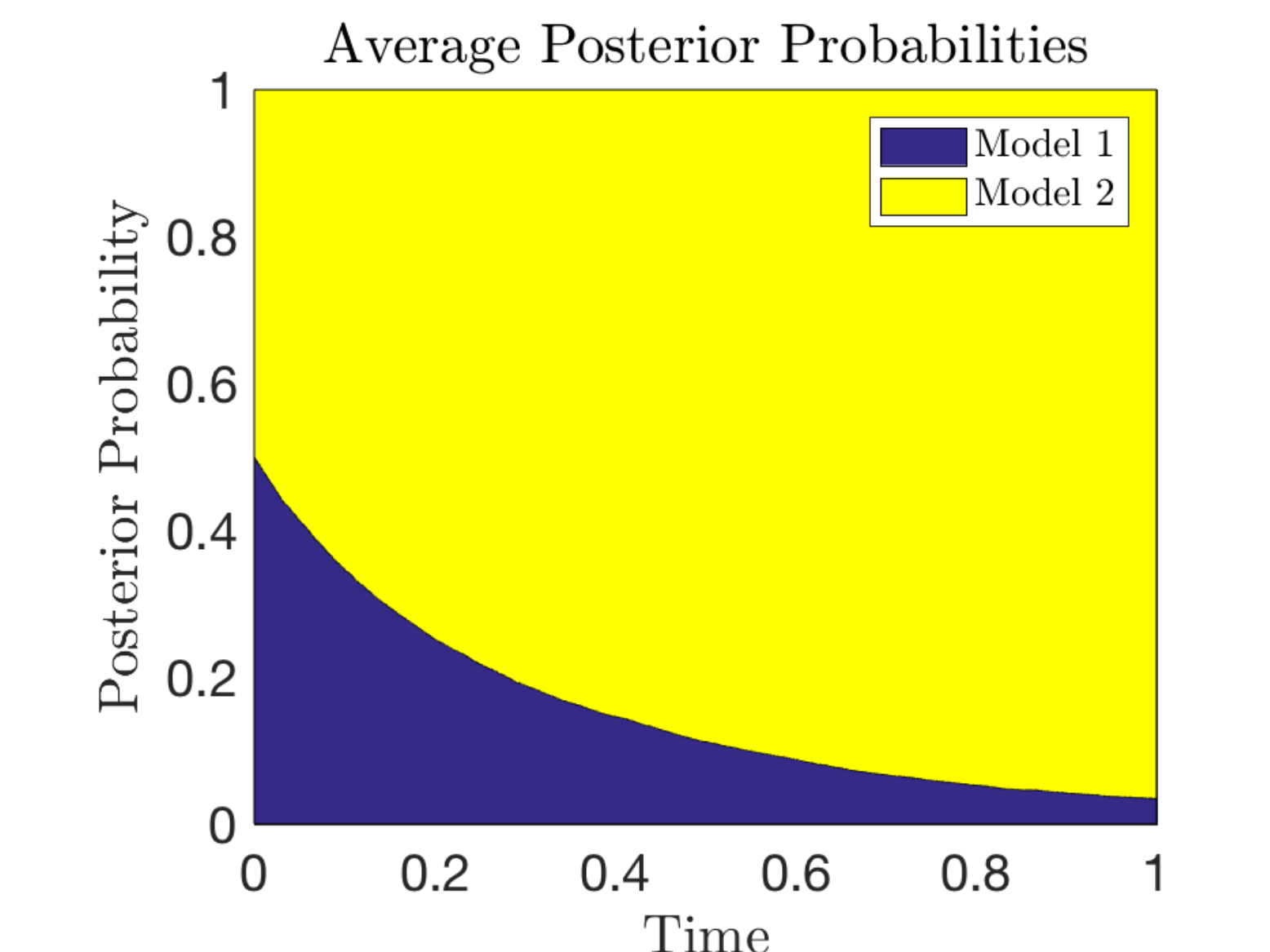}
    \end{subfigure}
    \\
    \centering
    \begin{subfigure}[t]{0.48\textwidth}
        \centering
        \includegraphics[width=\textwidth]{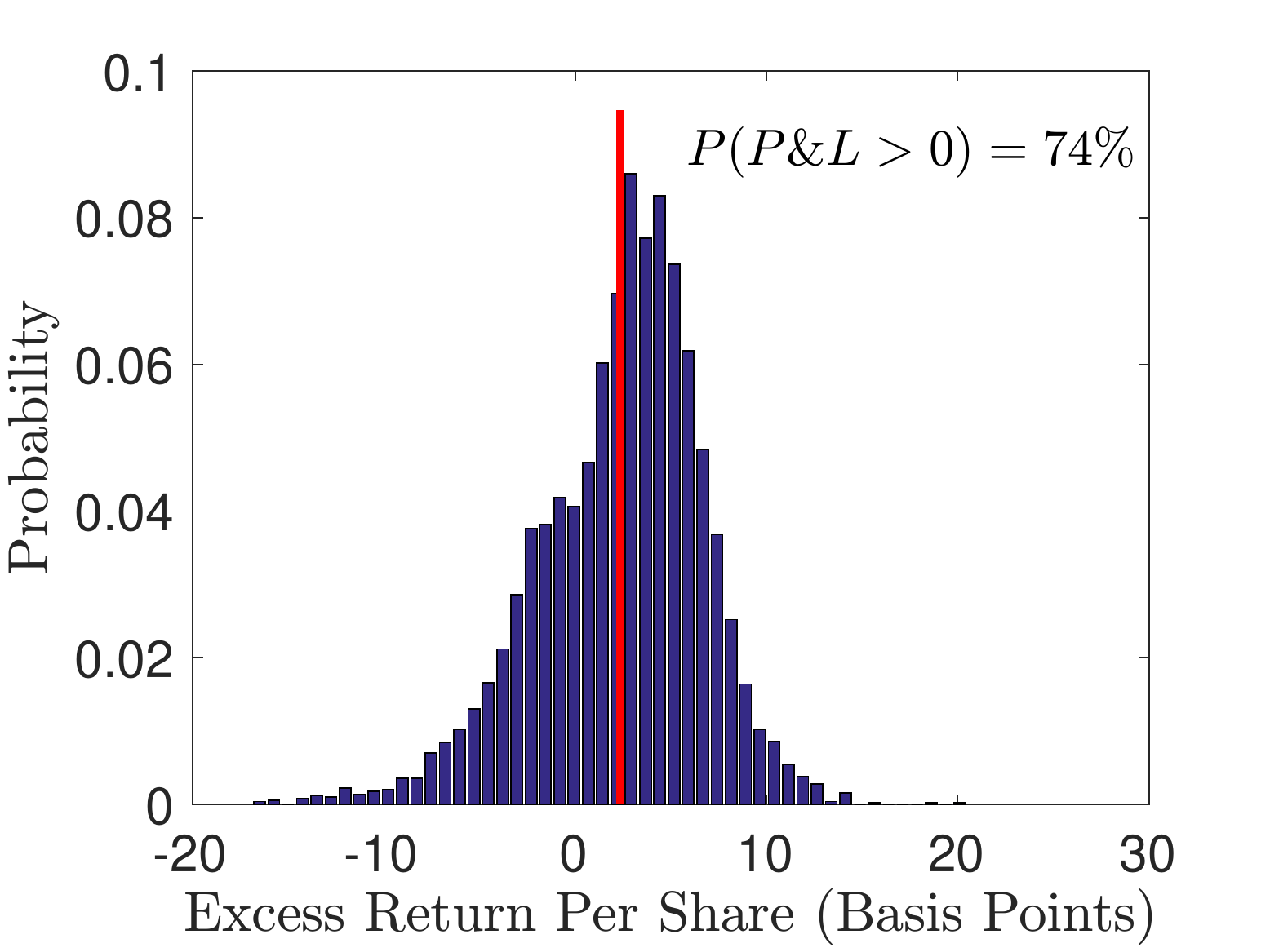}
    \end{subfigure}
     \begin{subfigure}[t]{0.48\textwidth}
        \centering
        \includegraphics[width=\textwidth]{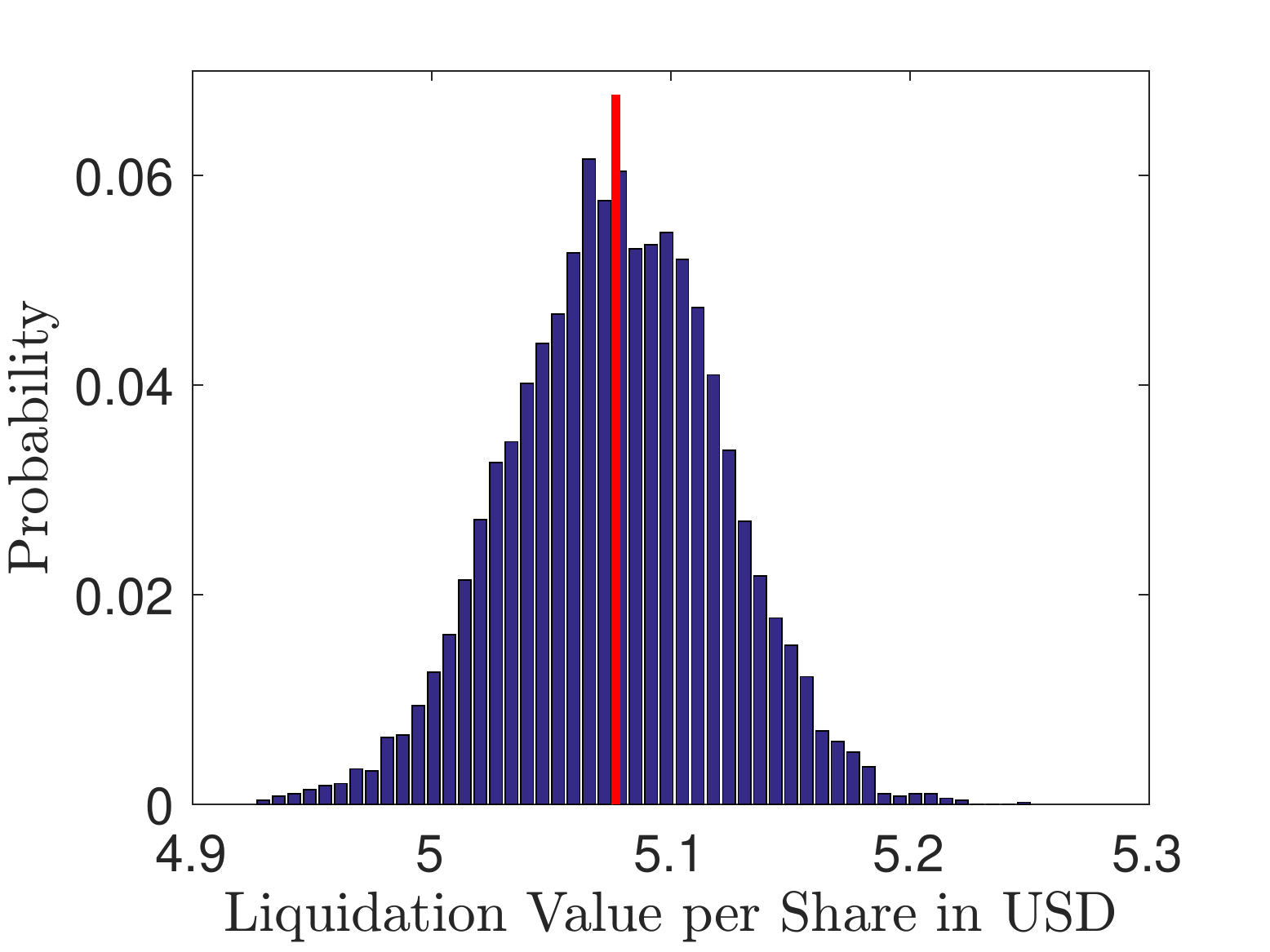}
    \end{subfigure}
    \caption{Simulation Results with an Ornstein-Uhlenbeck process}
    \label{fig: Simulation OU Results 1}
\end{figure}

Figure~\ref{fig: Simulation OU Results 1} shows the results of the simulation. The top right panel contains a heat map of the posterior probability of the two models. It shows, as time advances the trader on average will detect that the true rate of mean reversion is $\theta_2=\$ 5.15$. Moreover, by the end of the trading period, she is on average at least $97\%$ confident that model 2 is the true model governing asset prices. The top left graph in figure~\ref{fig: Simulation OU Results 1} shows a heat-map of the trading speed for the investor, where the dashed line represents the classical AC strategy. The dotted-and-dashed line represents the median of the traders's strategy. The heat-map shows how the trader adjusts her positions in a manner consistent with her predictions: as the investor discovers that $\Theta= \$ 5.15$, she expects the asset price to rise over time. Because of this she slows down her rate of liquidation initially, so that she can sell her asset at a higher price towards the end of the trading period. She then must speed up trading towards the end in order to unwind her position. The bottom left panel shows the histogram of the excess return per share of the optimal control over the AC control,
where the excess return is defined as
$
\frac{X_T^{\nu^\star}-X_T^{\nu^{\text{AC}}}}{ X_T^{\nu^{\text{AC}}}} \times 10^{4}
$
and $X_T^{\nu^{\text{AC}}}$ is the total cash the trader earns using the AC liquidation strategy. As the histogram shows, the filtered strategy outperforms the AC strategy during at least $73\%$ of the simulations. Lastly, the bottom right panel shows the trader's liquidation value per share over the trading period.

\begin{figure}[t!]
    \centering
        \includegraphics[width=0.31\textwidth]{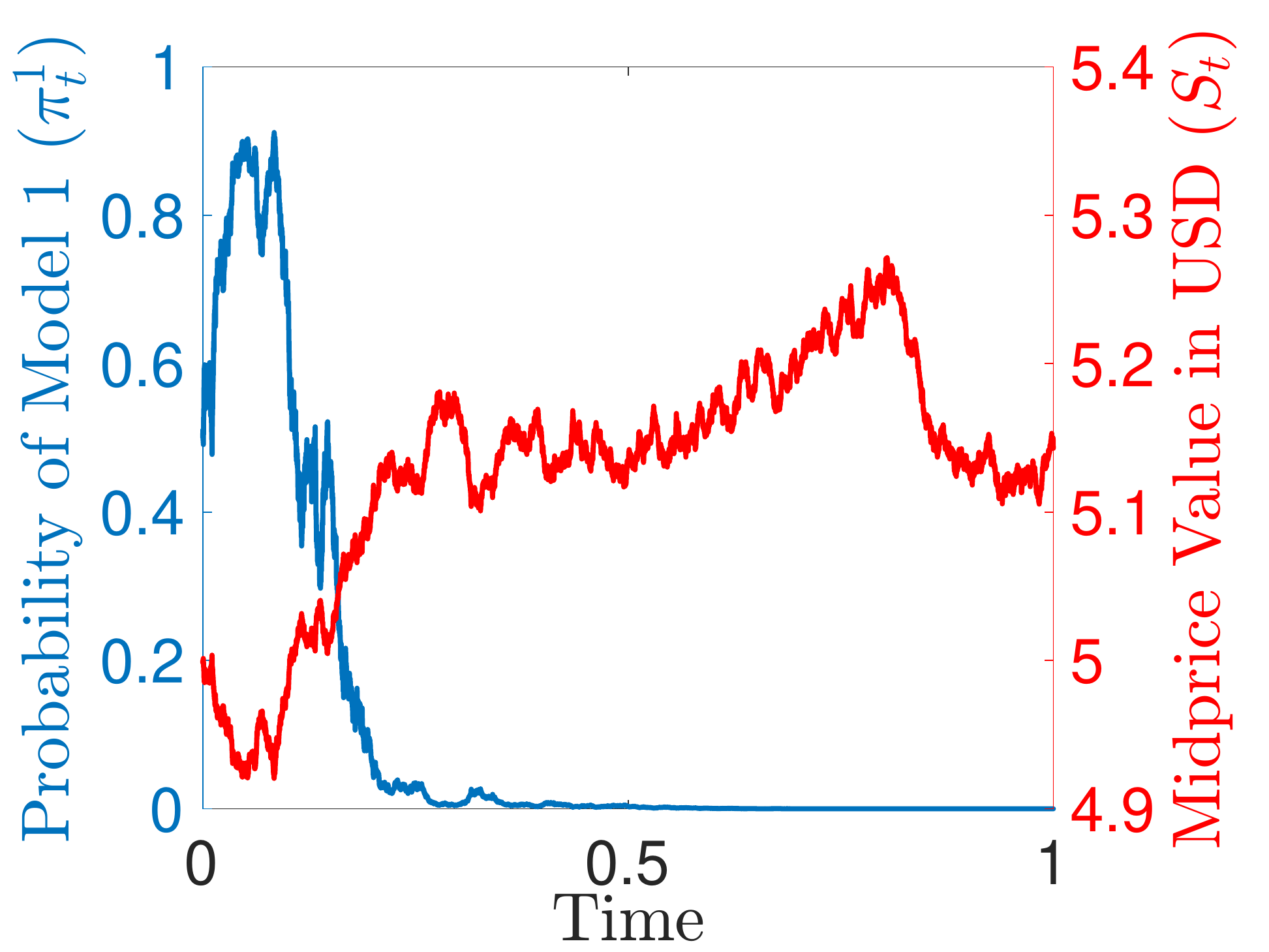}
    ~
        \includegraphics[width=0.31\textwidth]{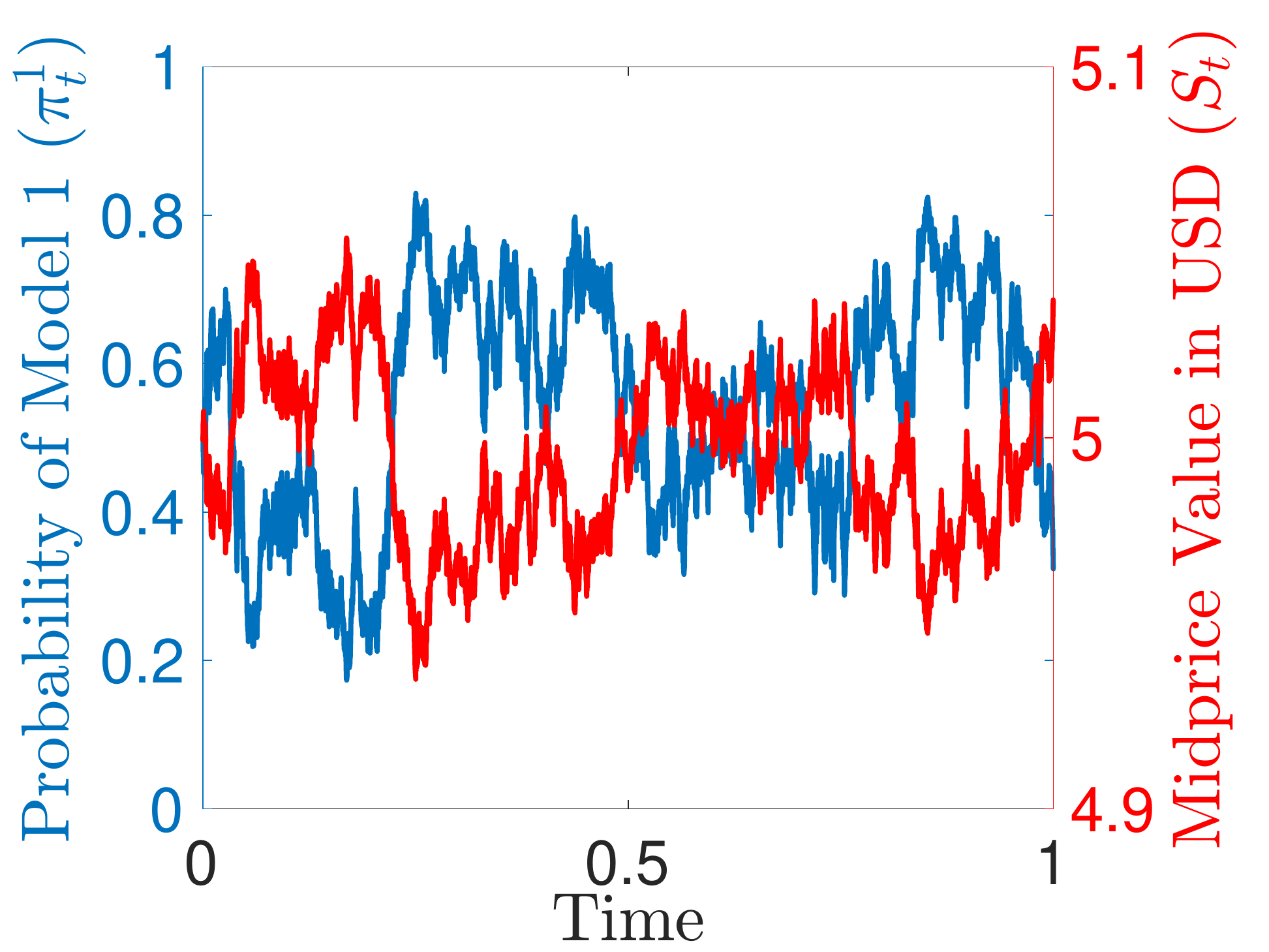}
    ~
        \centering
        \includegraphics[width=0.31\textwidth]{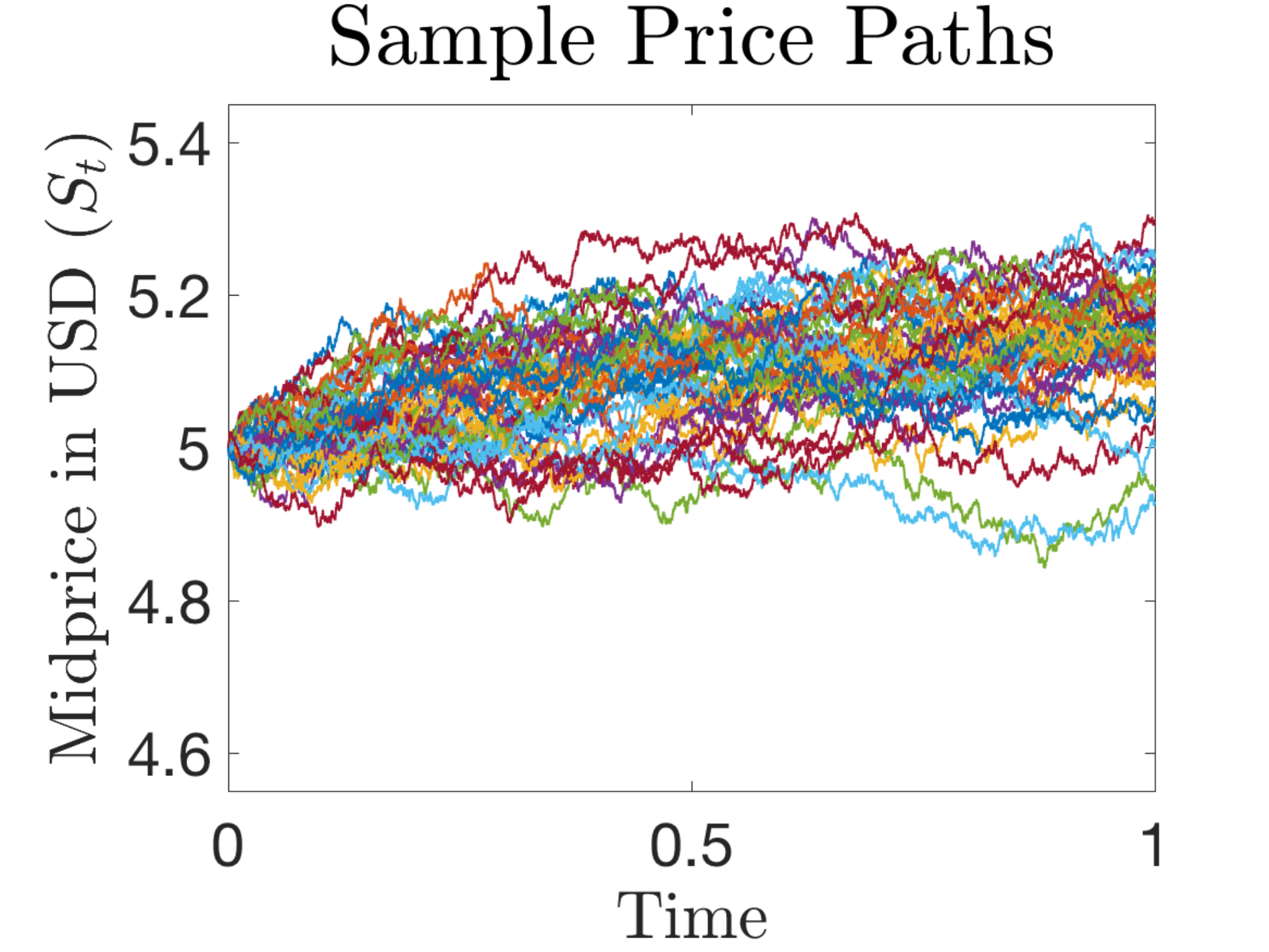}
    \caption{Sample simulation paths with an Ornstein-Uhlenbeck process.}
    \label{fig: Simulation OU Results 2}
\end{figure}
Figure~\ref{fig: Simulation OU Results 2} displays sample paths of the asset price and the filter. The top left plot demonstrates how the trader quickly detects the correct model based on the asset trajectory. In this simulation, the asset price initially drops, but then increases consistently. The posterior probability of model 1 adjusts accordingly, and initially rises, but then quickly drops and remains low. In the simulated path in the middle panel, the path of the midprice fluctuates around $\$5$ over the entire time period. The trader's estimate for the posteriori probability of model 1 varies according to the price movements she is observing. The resulting strategy induced by the filter fluctuating is an advantage to the trader, because the fluctuating filter more accurately reflects the actual behaviour of the asset price path -- as opposed to being certain that the true model is model 2. Finally, the bottom panel of the figure shows a collection of $40$ sample midprice path trajectories.

\subsection{Mean-Reverting Pure Jump Process} \label{sec: Second Example}

In this section, we investigate the case where the trader begins with no inventory and aims to gain profits from her alpha model through the use of a round-trip trading strategy. The asset price is assumed to be completely driven by the market order-flow, so that there is no diffusion or drift in the unaffected midprice. We assume the asset price mean-reverts to some unknown level $\Theta_t$ which the trader must detect. More specifically, the asset midprice in USD satisfies the SDE
\begin{equation}
  dF_t = b\;(dN_t^+ - dN_t^-)
  \;,
\end{equation}
where $N_t^+$ and $N_t^-$ are doubly stochastic Poisson processes with intensities $\lambda^+_t$ and $\lambda_t^-$ defined by
\begin{equation}
  \lambda_t^+ = \mu + \kappa \left( \Theta_t - F_t \right)_+
  \, \text{and } \,
  \lambda_t^- = \mu + \kappa \left( \Theta_t - F_t \right)_-
  \;,
\end{equation}
where $(x)_+$ and $(x)_-$ denote the positive and negative parts of $x$, respectively.

We assume $\Theta_t$ is a Markov chain with generator matrix $\bm C$, specified in Table~\ref{tbl:jump-params}. The filter for $\Theta_t$ cannot be computed explicitly, but it may be approximated via a Euler-Maruyama scheme of the SDE for the logarithm of the filter (see the SDE in Theorem~\ref{th: SDE Filter Theorem}). The resulting approximation for the value of the filter, given that the values of $\bN$ have been observed at times $\{t_k\}_{k=1}^K$, where $t_0=0$ and $t_K = T$ is obtained via the recursive formula
 \begin{subequations} \label{eq: SDE filter pure jump approximation}
   \begin{equation}
    \Lambda_{t_0}^j = \pi_0^j\;,
    \end{equation}
and
\begin{equation}
\begin{split}
     \Lambda_{t_{k+1}}^j =& \Lambda_{t_{k}}^j
     \exp\left\{
     2 \left(  1 - \mu - \tfrac{\kappa}{2} \left\lvert \theta_j - F_{t_k} \right\rvert \right) \Delta_{k+1}
     +
     \sum_{i=1}^J \left( \frac{\Lambda_{t_{k}}^i}{\Lambda_{t_{k}}^j} \right) C_{j,i} \;\Delta_{k+1}
    \right\}
        \\
    &\qquad
    \times \left( \mu + \kappa \left( \theta_j - F_{t_k} \right)_+ \right)^{\Delta N_{t_{k+1}}^+ }
    \times \left( \mu + \kappa \left( \theta_j - F_{t_k} \right)_- \right)^{\Delta N_{t_{k+1}}^-}\,,
\end{split}
\end{equation}%
 \end{subequations}%
$\forall k \geq 1$, where $\Delta N_{t_{k}}^\pm = N_{t_{k}}^\pm - N_{t_{k-1}}^\pm$ and $\Delta_k = t_k = t_{k-1}$. Alternatively, the trader's filter for $\Theta$ may be found using the forward equations, which are discussed in more detail in~\ref{sec: Forward Backward Section} as well as in Section~\ref{sec: EM Algorithm}. The forward equation approach is recommended if the time-steps $\Delta_k$ are relatively large, which introduces inaccuracies into the above Euler-Maruyama approximation to the solution of the filtering SDE of Theorem~\ref{th: SDE Filter Theorem}.

We assume the trader has a one hour trading horizon, and completely unwinds her positions by the end of the trading period. The optimal control in this set-up can be found in closed form. Let us define the constant $\kp = b \, \kappa$ and the $J\times J$ matrix $\bCp = \bm C + \kp  \eyeM $, where $\eyeM$ represents the $J \times J$ identity matrix. Let us also define the functions $\bm \Psi_1: [0,T]\times\Real^{J \times J} \rightarrow \Real^{J \times J}$, $\bm \Psi_2: [0,T]\times\Real^{J\times J} \rightarrow \Real^{J\times J}$, $\psi_1: [0,T]\times\Real \rightarrow \Real$ and $\psi_2: [0,T]\times\Real \rightarrow \Real$ where
\begin{align*}
  \bm \Psi_1(\tau,\; \bm Y)
  &= \left( e^{\tau \gamma} - e^{-\tau\gamma} \right)^{-1}
  e^{\tau  \bm Y}
  \left( \,
  \bm \Psi_2(\tau,\; \gamma \, \eyeM -  \bm Y )
  +
  \bm \Psi_2(\tau,\; -\gamma \, \eyeM -  \bm Y )
  \,
  \right)
  \;,
  \\
  \bm \Psi_2(\tau,\;\bm Y)
  &= \bm Y^{-1} \left( e^{\tau
  \bm Y} - \eyeM \right)
  \;,
\end{align*}
and where $\psi_1$ and $\psi_2$ are the scalar versions of the functions above, defined as
\begin{align*}
  \psi_1(\tau, \; y ) &=
  \left( e^{\tau \gamma} - e^{-\tau\gamma} \right)^{-1}
  e^{\tau y}
  \left( \,
  \psi_2(\tau,\; \gamma  -  y )
  +
  \psi_2(\tau,\; -\gamma -  y )
  \right)
  \, , \text{ and}
  \\
  \psi_2(\tau,\; y)
  &= y^{-1} \left( e^{\tau
  y} - 1 \right)
  \;.
\end{align*}
The optimal trading speed for this set-up, letting $\alpha \uparrow \infty$ is
\begin{align*}
  \nu^\star (t,F,Q,\bLambda)=  &-\gamma \coth (\gamma \, (T-t) )\,Q  \\
  &+
  \kp \Big [
  - F \;\psi_1(T-t, \, \kp )
  + \bpi^\T(\bLambda) \, \Psi_1(T-t, \, \bm C ) \, \bm \theta
  \\
  & \qquad\quad - \kp \bpi^\T(\bLambda) \, {\bCp}^{-1} \left(
  \Psi_1(T-t, \, \bCp)
  - \psi_1(T-t, \, \kp ) \eyeM
   \right) \bm \theta
  \Big]
  \;,
\end{align*}
where in the above, $\bm \theta = (\theta_j)_{j\in\mfJ}$ is a column vector containing all of the possible values that $\Theta_t$ can take.

In the numerical experiments, we assume two possible latent states for $\Theta_t$: $\theta_1 = \$ 4.9$ and $\theta_2= \$ 5.1$, and that the investor has an uninformed prior: $\pi_0^1=\pi_0^2=0.5$. The remaining parameters are provided in Table \ref{tbl:jump-params}. Note that we assume the latent process's generator matrix is symmetric so the trader has no a priori preference for the asset's price trajectory. The parameters for $\mu$, $\kappa$ and $\bm C$ are all taken to be in the same range as the calibrated parameters for the 2-state model fit to Intel Corporation (INTC) stock data. These calibrated parameters are found in Section~\ref{sec: Calibration Results}, and are discussed in more detail in Section~\ref{sec: EM Algorithm}. The parameters found in Section~\ref{sec: Calibration Results} are on a per-second scale and those here are on a per-hour scale, so we must multiply them by 3600 to get those found in Table~\ref{tbl:jump-params}.

 \begin{table}[h!]
  \centering
  \begin{tabular}{lllll}
  $\mfN=0$, & $F_0 = \$ 5$, & $b= \$ 0.01$
  & $\phi= 3 \times10^{-6},$ &
  $\bm C= \left[ \begin{smallmatrix} -10 & 10 \\ 10 & -10  \end{smallmatrix} \right]$, \\
  $\kappa = 1077$, & $\mu = 481$, & $a=\$ 10^{-5}$, & $\beta = \$ 10^{-3}$.
  \end{tabular}%
  \caption{The parameters in the pure jump mean-reverting model. All of the time-sensitive parameters are defined on an hourly scale. \label{tbl:jump-params}}
\end{table}

\begin{figure}[h!]
    \centering
    \begin{subfigure}[t]{0.48\textwidth}
        \centering
        \includegraphics[width=\textwidth]{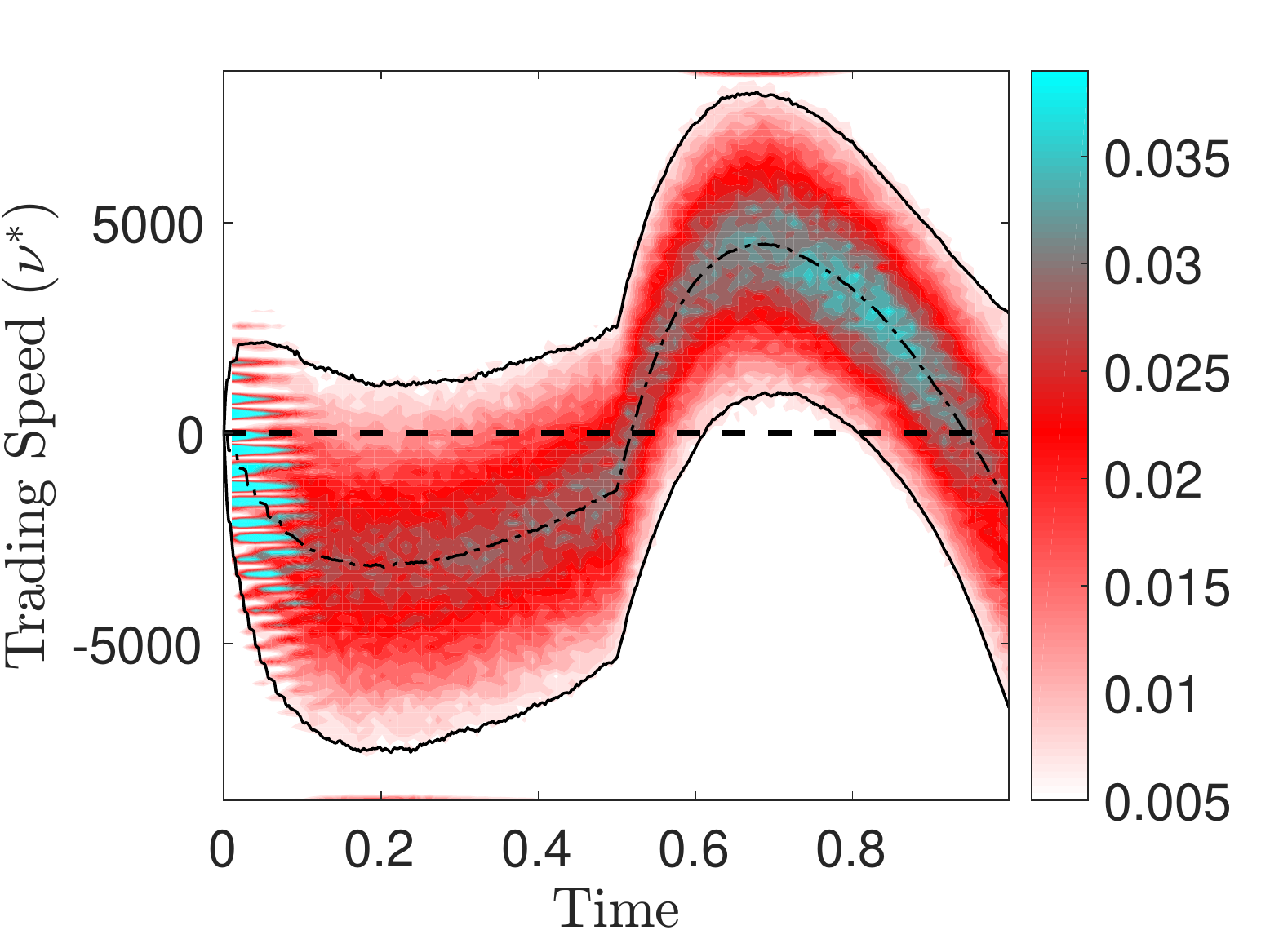}
    \end{subfigure}%
    ~
    \begin{subfigure}[t]{0.48\textwidth}
        \centering
        \includegraphics[width=\textwidth]{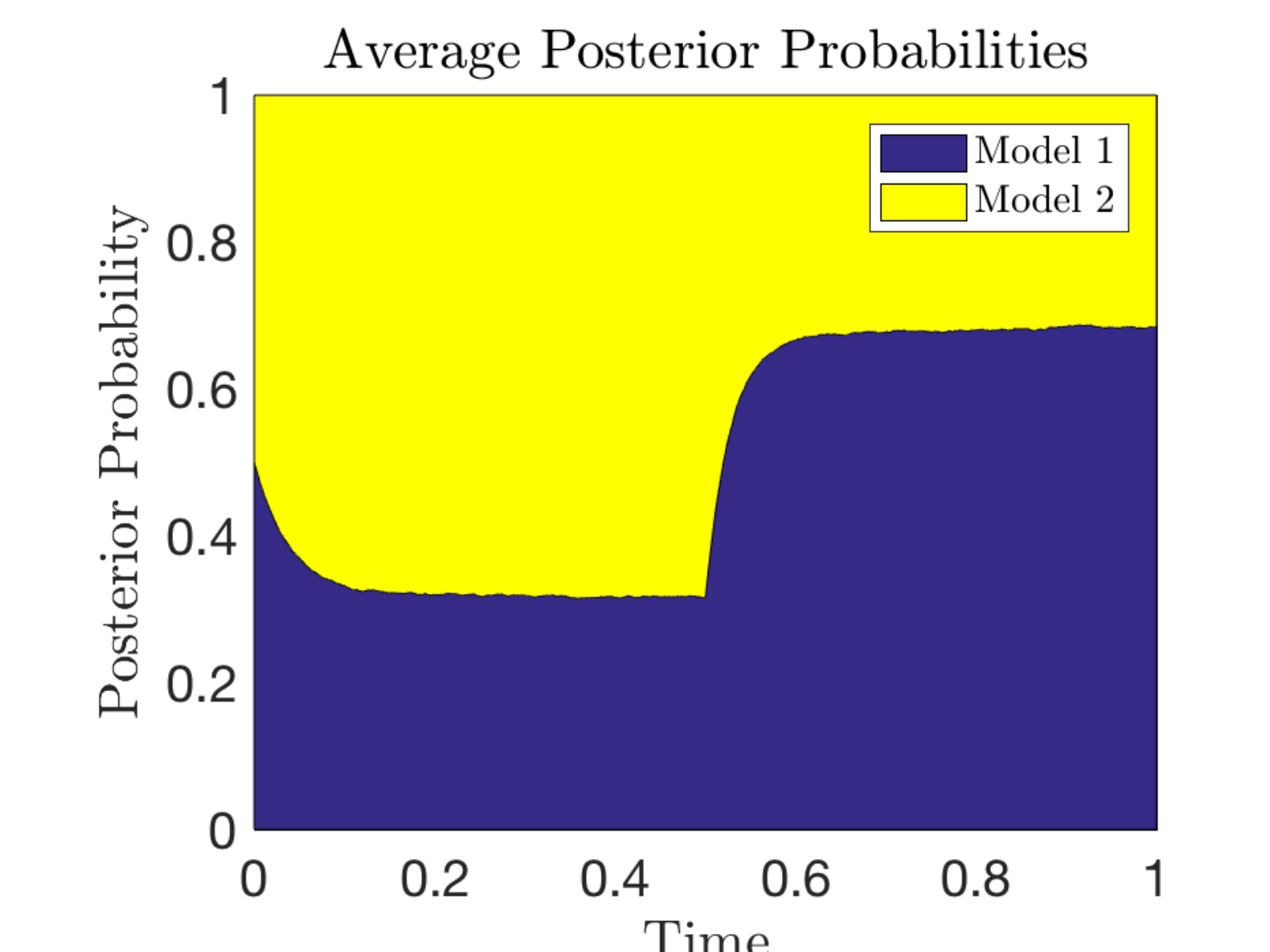}
    \end{subfigure}
    \\
    \begin{subfigure}[t]{0.48\textwidth}
        \centering
        \includegraphics[width=\textwidth]{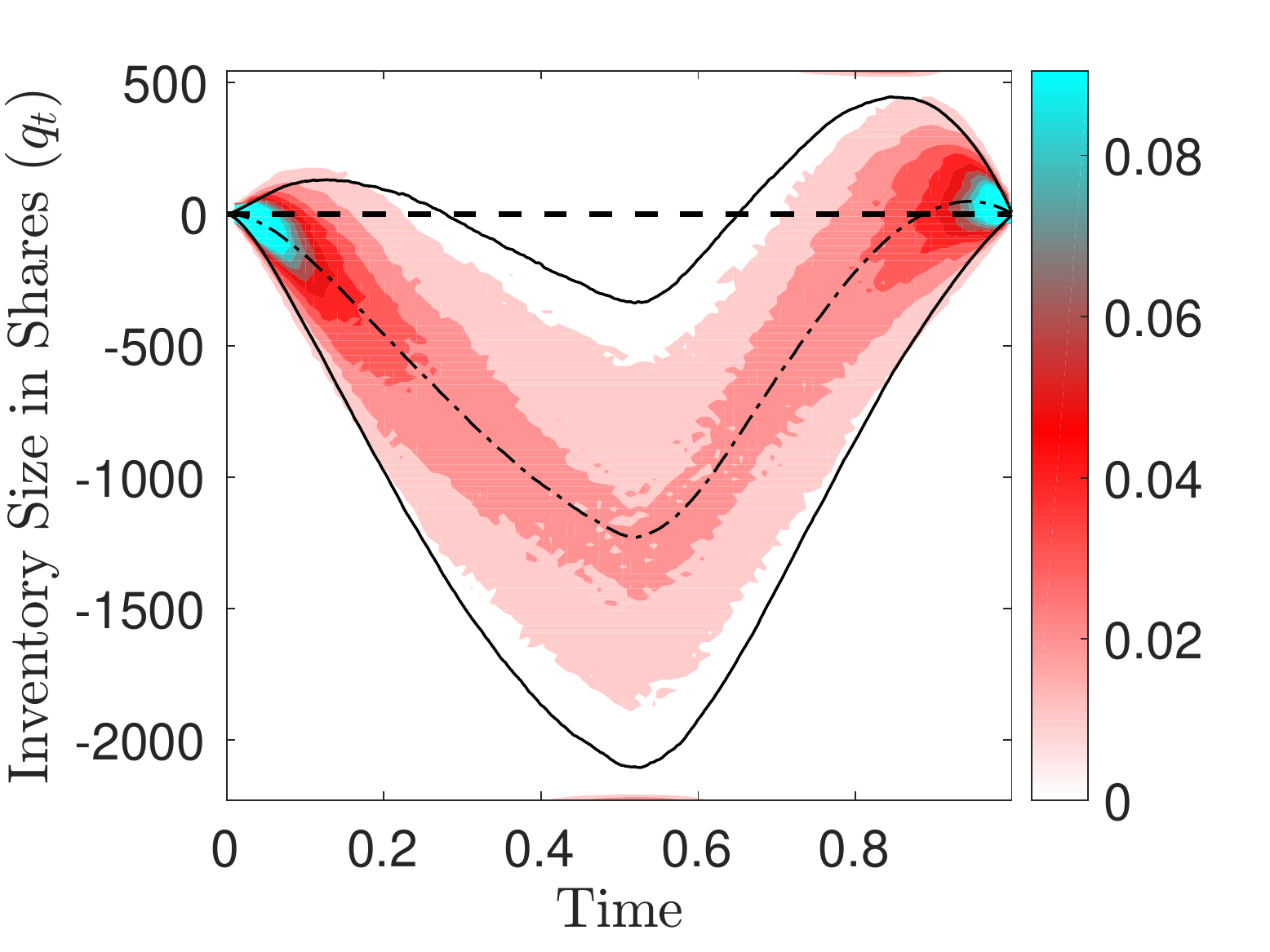}
    \end{subfigure}
	~
    \begin{subfigure}[t]{0.48\textwidth}
        \centering
        \includegraphics[width=\textwidth]{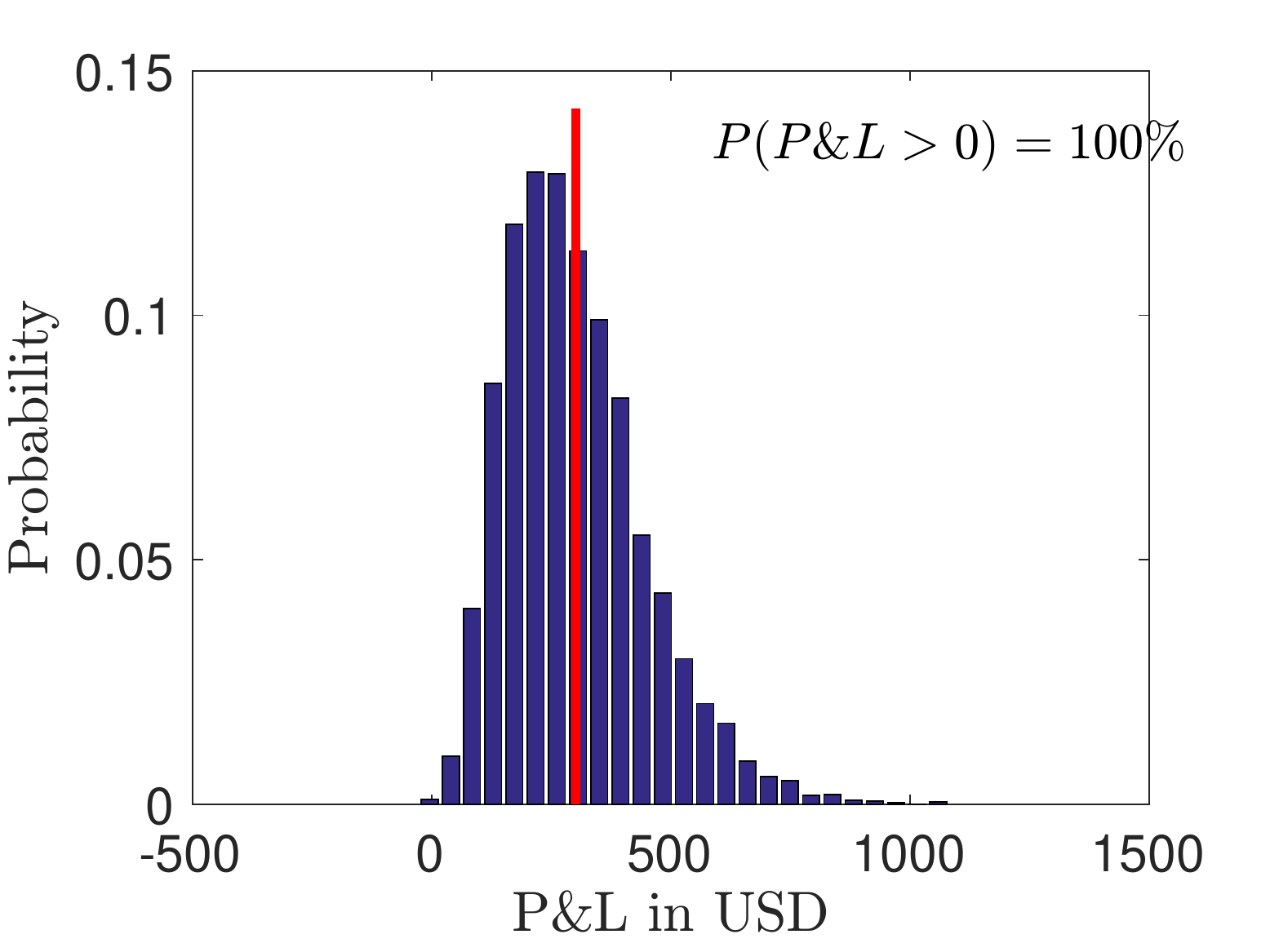}
    \end{subfigure}
    \caption{Simulation Results with a Pure-Jump Mean-Reverting Process}
    \label{fig: Simulation PureJump Results 1}
\end{figure}
For this simulation we fix a path for the latent process so that the value of $\Theta_t$ stays at $ \$ 5.1$ over the period $t\in[0,0.5]$ after which it jumps down to $\$ 4.9$ and remains there until the end of the trading period. This set-up will put the trader at a disadvantage. With the generator matrix defined in Table~\ref{tbl:jump-params}, the trader expects the latent process to jump an average of 10 times during the hour, whereas the path we fix for $\Theta$ jumps only once. Figure~\ref{fig: Simulation PureJump Results 1} shows the performance of the trader over the course of the hour. This implies that the trader will be acting based on the assumption that the latent process is on average 10 times more active than it is in the simultaion.

The top right portion of figure~\ref{fig: Simulation PureJump Results 1} demonstrates that on average, the filter detects the jump from $\Theta_t=\$ 5.1$ to $\Theta_t = \$ 4.9$. Furthermore, the bottom right panel shows the trader made a positive profit in $100\%$ of simulated scenarios. The left panels show heat maps of the trading speed (top left) and inventory (bottom left), respectively. Just as the strategy was constructed, the agent unwinds inventory by the end of the trading horizon. Moreover, she attains a short position initially, because she detects that $\Theta_t = \$ 5.1$ in the first half of the strategy, but expects the regime to switch causing a drop in asset price.

\section{Model Calibration}

This section shows how to calibrate the model in Section~\ref{sec: The Model} to market data using maximum likelihood estimation (MLE) through the expectation-maximization (EM) algorithm. This procedure will simultaneously enable us to classify hidden states within historical data and determine what kind of dynamics governs the behaviour of the midprice and order-flow within each of the hidden states. The EM algorithm is based on \cite{dempster1977maximum} and the Baum-Welch algorithm \citep{baum1970maximization}. Once the general procedure is described, we apply it to calibrate a generalized version of the pure-jump mean-reverting model in Section~\ref{sec: Second Example}.

\subsection{The EM Algorithm} \label{sec: EM Algorithm}

This section presents the discretized version of the asset and market dynamics in Section~\ref{sec: The Model} and its associated calibration algorithm. We assume we observed $D$ independent paths (e.g., from several different days of trading within the same trading hour) of the process $\left( F, \bN, \blambda \right)$ at discrete, uniformly spaced, times $\mathcal{T}=\left\{ t_k = \Dt \times k : k=0,\dots, K \right\}$, where $\Dt>0$ is the time interval in between observations.

Define the the discrete time processes $Y_k = \left( F_{t_k}, \bN_{t_k}, \blambda_{t_k} \right)$ and $Z_k = \Theta_{t_k}$ for each $t_k\in\mathcal{T}$. Given the assumptions in Section~\ref{sec: The Model},  $(Y,Z)=\left\{ ( Y_k, Z_k ) \right\}_{k=0 \dots K}$ is a Markov process. Furthermore, the transition matrix of the process $Z_k$ and the distribution of $Z_0$ are known and can be written as
\begin{align}
\bm P &= \left[ \mathbb{P} \left( Z_k = \theta_j \mid Z_{k-1} = \theta_i  \right) \right]_{i,j=1}^J
= e^{\Dt \bm C}, \qquad \text{and} \\
\bm{\pi}_0 &= \left[ \mathbb{P} \left( Z_0 = \theta_i  \right) \right]_{i=1}^J
= \left( \pi_0^i \right)_{ {i=1} }^J \,.
\end{align}
Moreover, by the definitions of the processes $F$, $N^+$, $N^-$ and $\Theta$, the processes $Y$ and $Z$ inherit the conditional independence structure depicted in the directed graph shown in Figure \ref{fig:HMM-graph}. It is important to note that here we allow for dependence between the visible states $Y$ even when conditioned on the latent states (i.e., there are connections between subsequent $Y$). This differs from the usual HMM setup where $Y$ are conditionally independent when conditioned on $Z$.
\tikzstyle{state}=[shape=circle,draw=blue!50,fill=blue!20]
\tikzstyle{observation}=[shape=rectangle,draw=orange!50,fill=orange!20]
\tikzstyle{lightedge}=[<-,dotted]
\tikzstyle{mainstate}=[state,thick]
\tikzstyle{mainedge}=[<-,thick]

\begin{figure}[h!]
\begin{center}
\begin{tikzpicture}[]
\node[state] (s1) at (0,2) {$Z_0$};

\node[state] (s2) at (2,2) {$Z_1$}
    edge [<-] node[auto,swap] {$\boldsymbol{P}$} (s1);

\node[state] (s3) at (4,2) {$Z_2$}
    edge [<-] node[auto,swap] {$\boldsymbol{P}$} (s2);

\node[state] (s4) at (6,2)  {$Z_3$}
    edge [<-] node[auto,swap] {$\boldsymbol{P}$} (s3);

\node[observation] (y1) at (0,0) {$Y_0$};

\node[observation] (y2) at (2,0) {$Y_1$}
    edge [<-] node[align=left,   below] {$f_\psi$} (s1)
    edge [<-] (y1);
    
\node[observation] (y3) at (4,0) {$Y_2$}
edge [<-] node[align=left,   below] {$f_\psi$} (s2)
edge [<-] (y2);

\node[observation] (y4) at (6,0) {$Y_3$}
edge [<-] node[align=left,   below] {$f_\psi$} (s3)
edge [<-] (y3);

\end{tikzpicture}
\end{center}
\caption{Directed graphical representation of the price and order-flow model. \label{fig:HMM-graph}}
\end{figure}
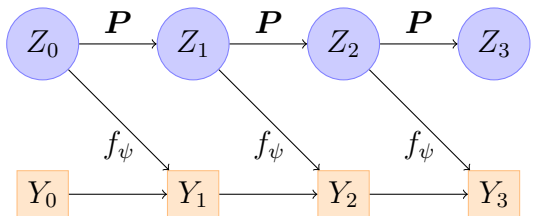

More specifically, we have
\begin{equation}
	\P \left( \, Y_k \in \mcB , Z_k = \theta_i \mid Y_{k-1}, Z_{k-1} \, \right)
	= \P \left( Y_{k} \in \mcB \mid Y_{k-1}, Z_{k-1} \, \right)
	\P \left( Z_{k} \mid Z_{k-1} \right)
	\;,
\end{equation}
for eack $k=1\dots K$ and any Borel measurable set $\mcB$. Let us also assume the transition density for $Y_k$ conditional on $(Y_{k-1},Z_{k-1})$ is known for each $k=1,\dots,K$, and that given a set of suitable parameters $\psi$ we can write
\begin{equation}
	\mathbb{P}\left( Y_{k} \in dy \mid Y_{k-1} = y_0 ,Z_{k-1} = \theta_i \right) =
	f_\psi (t_{k},\,y\, ; t_{k-1}, y_0 , \theta_i ) \, \bmu(dy) ,
\end{equation}
where for each fixed $(y_0,\theta_j)\in\mathds{R}^{2J+3} \bigtimes \left\{ \theta_j \right\}_{j\in\mfJ}$, and fixed times $0\leq t_{k-1} < t_{k} \leq T$, $f_\psi$ is a probability density function, and where $\bmu$ is the Lebesgue measure. In machine learning language, the function $f_\psi$ is often called the \textit{emission probability}. It may be found by solving the appropriate Kolmogorov/Fokker-Planck forward equation. If we consider $f^i_\psi(t,y) = f_\psi (t,\,y\, ; s , y_0 , \theta_i )$ for fixed values of $y_0$ and for $0\leq s\leq t\leq T$, then $f^i_\psi$ is the solution to the PDE
\begin{equation} \label{eq: Kolmogorov Backward}
\begin{cases}
	(\, \partial_t + \sum_{j\in\mfJ} \left(e^{t\,\bm{C}}\right)_{i,j} \widehat{\mathcal{L}}^\dag_j \;)f^i_\psi (t,y) = 0 \\
	f^i_\psi (s,y) = \delta(y-y_0)\,,
\end{cases} \;,
\end{equation}
where $\delta(y)$ is the dirac delta function, and the operator $\widehat{\mathcal{L}}^\dag_j$ is the adjoint infinitesimal generator of the process $Y_t = \left( F_t, \bN_t, \blambda_t \right)$ conditional on the event $\{ \omega : Z_t(\omega) = \theta_j , \forall t\in[0,T] \}$, and where $\left(e^{t\,\bm{C}}\right)_{i,j}$ is element $(i,j)$ of the matrix exponential of $t\, \bm C$. In practice, it is possible to use any number of approximations to obtain the emission probability density.

In the remainder of this section we use the notation $Y_k^d$ and $Z_k^d$ to denote the processes $Y$ and $Z$, observed at time $t_k$ in the $d^{\text{th}}$ independently observed path. We also introduce the notation $\mathcal{Y}_{m:n}^d \triangleq \bigcap_{m\leq k \leq n} \left\{ \omega: Y_k^d \in dy_k^d \right\}$ where each $\left\{ \omega: Y_k^d \in dy_k^d \right\}$, represents the event in which we observe the visible process $Y_k^d$ having taken the value $y_k^d$. 


The objective is to find the set of parameters $\Gamma = ( \pi_0 , \bm P , \psi )\in\mathfrak{G}$ which maximize the likelihood of having made observations $\{ y_k^d \}_{k,d=1}^{K,D}$ of the process $Y$, where $\mathfrak{G}$ is the set of allowable parameters, defined as
\begin{equation}
\begin{split}
	\mathfrak{G} =&
	\left\{ \bm\pi = \left( \pi_i \right)_{ {t\in[0,T]} }^J \in \mathds{R}^J \,:\, \bm{1}_{J}^\T\bm\pi = 1  , \pi_i \in [0,1] \right\}
	\\  & \quad \bigtimes
	\left\{ \bm P = \left( P_{i,j} \right)_{ {t\in[0,T]} }^J \in \mathds{R}^{J\times J} \,:\,
	\bm{1}^\T_J \bm P = \bm{1}_J , P_{i,j} \in [0,1]
	\right\}
	\bigtimes \mathfrak{G}_\psi
	\;,
\end{split}
\end{equation}
where $\bm{1}_J$ is a column vector of ones, and $\mathfrak{G}_{\psi}$ is the set on which we restrict $\psi$. The first two sets in the definition of $\mathfrak G$ ensure the entries of $\bm \pi_0$ sum to one, and $\bm P$ is a valid Markov chain transition matrix, respectively.

Due to the presence of the unobserved (latent) states $Z$, the number of terms in the log-likelihood of the discrete observations of the process $Y$ grows exponentially with the number of observations. It is, therefore, necessary to use the EM algorithm which provides a sequence of improving parameter estimates $\{ \Gamma_{k} \}_{k=1}^{\infty}$. Each term in the sequence has larger likelihood than the previous.

The EM algorithm proceeds as follows. Begin with an initial guess (e.g., as estimated from an equivalent model with no latent states) for the model parameters, $\Gamma_0 \in \mathfrak{G}$. Generate a recursive sequence of parameters $\{\Gamma_n\}_{n=1}^\infty$, defined by the relationship
\begin{equation} \label{eq: EM Algorithm Sequence}
  \Gamma_{n+1} = \arg \sup_{\Gamma \in \mathfrak G}\; \mathbb{E}^{\mathbb{P}^{\Gamma_{n}}}
  \left[ \; \left. \log L^\Gamma \; \right| \; \{ \mathcal{Y}_{0:K}^d\}_{d=1}^D \right]
  \;,
\end{equation}
where $\log L^\Gamma$ is the joint log-likelihood of having observed $D$ independent paths of the process $(Y,Z)$ given the parameter set $\Gamma$, and where $\mathbb{P}^{\Gamma_n}$ is the probability measure conditional on the dynamics of $(Y,Z)$ having parameters $\Gamma_n$. The joint log-likelihood for our model yields the decomposition
\begin{equation}
\begin{split}
 \log L^\Gamma =&\phantom{+} \sum_{d=1}^D \sum_{i=1}^J \log\left( \pi_0^i \right) \, \1{\Znd{0}{i}} \\
  &+ \sum_{d=1}^D \sum_{k=0}^{K-2} \sum_{i,j=1}^J \log\left( P_{i,j} \right)
  \, \1{\Znd{k}{i} , \, \Znd{k+1}{j}} \\
  &+ \sum_{d=1}^D \sum_{k=0}^{K-1} \sum_{i=1}^J \log \left( f_\psi(t_{k+1},y_{k+1}^d ; t_k, y_{k}^d,\theta_i) \right) \, \1{\Znd{k}{i}}
  \;.
\end{split}
\end{equation}
We then have
\begin{equation} \label{eq: EM algorithm E step}
\begin{aligned}
	\mathbb{E}^{\mathbb{P}^{\Gamma_{n}}}
  \left[ \log L^\Gamma \mid \mathcal{Y}_{0:K}^1 , \dots, \mathcal{Y}_{0:K}^D  \right] = &
   \phantom{+}\sum_{d=1}^D \sum_{i=1}^J \log\left( \pi_0^i \right) \, \gamma_0^{i,d} \\
  &+ \sum_{d=1}^D \sum_{k=0}^{K-2} \sum_{i,j=1}^J \log\left( P_{i,j} \right) \, \xi_{k}^{i,j,d} \\
  &+ \sum_{d=1}^D \sum_{k=0}^{K-1} \sum_{i=1}^J \log\left( f_\psi(t_{k+1},y_{k+1}^d ; t_k, y_{k}^d,\theta_i) \right) \, \gamma_{k}^{i,d}
  \;,
\end{aligned}
\end{equation}
where   the \textit{smoother} $\gamma_k^{i,d}$ and \textit{two-slice marginal} $\xi_k^{i,j,d}$ are defined as
\begin{align}
	\gamma_k^{i,d} = \P^{\Gamma_n} \!\left( \Znd{k}{i} \,| \,\Y_{0:K}^d \right)\; \quad \text{ and } \quad
	\xi_k^{i,j,d} = \P^{\Gamma_n} \!\left( \Znd{k}{i}, \Znd{k+1}{j} \,| \,\Y_{0:K}^d \right)
	\;.
\end{align}
These coefficients can be computed by using the \textit{forward-backward algorithm}, which, for our model, is provided in~\ref{sec: Forward Backward Section}. Next, because each line of equation~\eqref{eq: EM algorithm E step} depends only on one of $\bm\pi_0$, $\bm P$, and $\psi$, the updated estimates for the parameters can be obtained independently. The resulting update rules (which maximize \eqref{eq: EM algorithm E step}) are
\begin{subequations}
\begin{align}
	\pi_0^{j\,\star} &= \tfrac{1}{D}\sum_{d=1}^D \gamma_0^{j,d},&  j\in\mfJ,
	\\
	P_{i,j}^\star &= \frac{\sum_{d=1}^D \sum_{k=0}^{K-2} \xi_{k}^{i,j,d} }{
	\sum_{d=1}^D \sum_{k=0}^{K-2} \sum_{j=1}^J \xi_{k}^{i,j,d}
	},   & i,j\in\mfJ,
	\\
	\psi^\star &= \arg \max_{\psi\in\mathfrak{G}_\psi}
			\left\{ \sum_{d=1}^D \sum_{k=0}^{K-1} \sum_{i=1}^J \log\left( f_{\psi}(t_{k+1},y_{k+1}^d ; t_k, y_{k}^d,\theta_i) \right) \, \gamma_{k}^{i,d} \right\}. \label{eqn:psi-update}&
\end{align}%
\end{subequations}%
Hence, $\Gamma_{n+1} = \left( \bm\pi_0^\star , \bm P^\star, \psi^\star \right)$. The updated estimates for the emission probabilities \eqref{eqn:psi-update}, may in some models be analytically tractable (e.g., gaussian mixtures, or discrete distributions), while in others one may have to resort to numerical optimization schemes.

More details on the maximization of equation~\eqref{eq: EM Algorithm Sequence} and of various other related computational concerns are outlined in~\cite{rabiner1989tutorial}. The one subtle difference between the model presented in this section and the classical Hidden Markov Model, such as the one found in~\cite{rabiner1989tutorial} is that the visible process $Y$ exhibits temporal correlation even when conditioning on the path of the hidden process $Z$.

\subsection{Mean-Reverting Pure-Jump Model} \label{sec: Censored Pure Jump}

In this section, we show how to implement the EM algorithm for the model presented in section \ref{sec: Second Example}. Assuming that the observations are frequent enough to observe one or no jumps in the midprice, we adopt a censored version of the model. As in Section~\ref{sec: Second Example}, we assume the increments of the midprice $F_t$ in the interval $[t_n,t_{n+1})$ satisfy
\begin{equation}
  F_{t_{n+1}} - F_{t_n} = b \left( \Delta N_{t_n}^+ - \Delta N_{t_n}^{-} \right),
\end{equation}
where each of the $\Delta N_t^{\pm} \in \left\{0,1\right\}$ are censored, conditionally independent Poisson random variables with respective stochastic rate parameters $\lambda_t^+\Dt $ and $\lambda_t^+\Dt$, where
\begin{equation}
  \lambda_t^\pm = \sum_{j=1}^J \lambda_t^{\pm, j} \1{\Theta_t^j = \theta_j} \;,
\end{equation}
and we assume $\lambda_t^{\pm,j}$ have constant paths over each observation period (this can easily be relaxed). Moreover,
\begin{equation}
  \lambda_{t_n}^{+,j} = \mu_j + \kappa_j \left( \theta_j - F_{t_n} \right)_+ \quad \text{and} \quad
  \lambda_{t_n}^{-,j} = \mu_j + \kappa_j \left( \theta_j - F_{t_n} \right)_-
  \;.
\end{equation}

This generalizes the model in Section~\ref{sec: Second Example} as all parameters $\mu$, $\theta$ and $\kappa$ may vary according to the state of the hidden process $\Theta_t$. The emission probabilities have parameters $\psi = \left\{ \mu_j\, \kappa_j, \theta_j \right\}_{j=1}^J $, which represent (within each state) the base noise level, the mean reversion rate, and the mean reversion level of the mid-price, respectively. With these ingredients, we can now state the emission probability as follows
\begin{equation}
\begin{split}
& f_\psi( t_{k+1} , Y_{k+1} ; t_{k} , Y_{k} , \theta_j  ) \\
& =
  \begin{cases}
    e^{- \Dt \lambda^{-,j}_{t_k}} \left( 1- e^{- \Dt \lambda^{+,j}_{t_k}} \right)
    ,& \text{ if } F_{t_{n+1}} > F_{t_k} \\
    e^{- \Dt \lambda^{+,j}_{t_k}} \left( 1- e^{- \Dt \lambda^{-,j}_{t_k}} \right)
    ,& \text{ if } F_{t_{n+1}} < F_{t_k} \\
    e^{- \Dt \left( \lambda^{+,j}_{t_k} + \lambda^{-,j}_t \right) }
    + \left( 1- e^{- \Dt \lambda^{+,j}_{t_k}} \right) \left( 1- e^{- \Dt \lambda^{-,j}_{t_k}} \right) ,& \;\text{otherwise.}
  \end{cases}
\end{split}
\end{equation}

Armed with this expression, we can apply the EM algorithm from Section~\ref{sec: EM Algorithm} to obtain parameter estimates. In this model, we cannot obtain explicit updates for $\psi$ from \eqref{eqn:psi-update}, and instead use a numerical optimization.

\subsection{Example Fit to INTC stock Data} \label{sec: INTC Fit Discussion}

In this section, we fit the model presented in Section~\ref{sec: Censored Pure Jump} to intraday price data. We fit this model to the midprice of INTC stock on the NASDAQ exchange taken at each second between the hours of 10:00 and 11:00 on each business day during 2014. To normalize this data between days, we subtract the price at 10:00 on the same day from each data point, so that we are fitting the model from Section~\ref{sec: Censored Pure Jump} to the price change in INTC since 10:00 on each day. Moreover, we assume that price paths are independent across days. From this dataset, we set $\Dt=1$, so that all of the parameter estimates can be interpreted on a per-second scale. In \ref{sec: Calibration Results} we show the parameter estimates for the model with $1$ to $6$ latent states. All of these parameter estimates were obtained using the EM algorithm described in Section~\ref{sec: EM Algorithm} applied to the censored pure jump model in Section~\ref{sec: Censored Pure Jump}.

An important part of calibrating models with latent states is to estimate the possible number of states that a model should have. To choose an `optimal' number of latent states we use two information criteria: the Bayesian Information Criterion (BIC), and the integrated completed likelihood (ICL). For the purposes of our paper, we use an approximation to the ICL which can be found in~\cite{biernacki2000assessing}. The definitions of both of these criteria and some related computational details can be found in~\ref{sec: Calibration Results}. It is well known that ICL typically under-estimates the true number of latent states, whereas the BIC typically over estimates the number of latent states. For INTC, BIC and ICL estimate 5 and 1 latent states, respectively. This implies that the true number of of states should be somewhere within that range.

Recall that $\mu_i$, $\kappa_i$ and $\theta_i$ can be interpreted as the noise level, the mean-reversion rate, and the mean-reversion level, respectively. We find that the mean-reversion level does not vary much (all less than $3\cent$ in absolute value) in comparison to realized daily stock price movements which can range anywhere in between $-50\cent$ to $50\cent$. This indicates the daily asset prices typically mean-revert back to their initial values no matter the state of the latent process. The distinction between states occurs in the strength of the mean-reversion through $\kappa_i$, and the base noise level $\mu_i$.

The estimates for the transition probabilities shown in~\ref{sec: Calibration Results} indicate there is persistence in the latent process regardless of the number of latent states. This implies the trader is able to detect market states and to act on them before the states switch again. From the estimated generator matrices in~\ref{sec: Calibration Results}, the mean time spent in a given state can range from a few seconds to over $800$ seconds, depending on the number of allowed latent states.

\section{Conclusion}

In this paper, we solved a finite horizon optimal trading problem in which the midprice of the asset contains a latent alpha component stemming from a diffusive drift as well as a pure jump component. We obtain the solution in closed form, up to the computation of an expectation which depends on the class of potential models. The optimal trading speed is found to be a combination of the classical AC trading strategy plus a modulating term that incorporates the trader's estimate of the latent factors and its forecast. The form of the solution is similar in spirit to the results in ~\cite{cartea2016incorporating} where the authors have an alpha component stemming from market order-flow, but in that work viewed as visible. We presented two examples one where the trader wishes to completely liquidate a large position (the optimal execution problem) and the other where the trader uses the latent states to generate a statistical arbitrage trading strategy.  Both examples show that there is significant value in incorporating latent states. Finally, we presented a method for obtaining parameter estimates from data using the EM algorithm and applied it to an example stock.

There are many potential future directions left open for investigation. One direction which we have already been investigated is to generalize the analysis to trading multiple assets, as well as incorporating multiple latent factors. In this work, the trader is assumed to execute trades continuously and uses market orders which walk the limit order book (LOB) and hence obtain a temporary price impact. It would be interesting to analyze the case of executing market orders at discrete times, and hence recast the problem as an impulse control problem with latent alpha factors. Along similar lines, the agent may wish to use limit orders to squeeze even more profits out of the strategy. Combining market and limit orders, along the lines of \cite{CarteaJaimungal2014} and \cite{Huitema}, but including latent alpha factors would also be quite interesting.

\newpage

\appendix

\footnotesize

\section{Proof of Theorem~\ref{th: SDE Filter Theorem}}\label{sec: Proof of SDE Filter Theorem}
\begin{proof}
  We present here the proof of the case where $\sigma>0$. The case where $\sigma=0$ and $A:=0$ can be derived in the same fashion by excluding all of the diffusive terms.

  Due to the Novikov Condition~\eqref{eq:NovikovCondition}, we may define the measure $\mathbb Q$ through the Radon-Nikodym derivative
\begin{align*}
  \left( \frac{d\mathbb{Q}}{d\mathbb{P}}\right)_t &=
  \exp\left\{ -\sigma^{-1} \int_0^t A_{u-} \, dW_u - \frac{\sigma^{-2}}{2}\int_0^t \left( A_{u-} \right)^2 du  \right\} \\
  &\times
  \exp\left\{
  \int_0^t (\lambda_{u-}^+ -1) du - \int_0^t \log(\lambda_{u-}^+ ) dN^+_u
  \right\} \\
  &\times
  \exp\left\{
  \int_0^t (\lambda_{u-}^- -1) du - \int_0^t \log(\lambda_{u-}^- ) dN^-_u
  \right\}
  \;,
\end{align*}
which is defined so that under the measure $\mathbb{Q}$, the process $ \sigma^{-1} (F_t - b (N_t^+ - N_t^-) ) $ is a Brownian Motion and both $N_t^{+}$ and $N_t^-$ have intensity process equal to 1. This measure is deliberately chosen so that $F_t$, $N_t^+$ and $N_t^-$ are $\mathbb Q$-independent of $\Theta_t$. Additionally, the inverse of this Radon-Nikodym derivative is
\begin{equation}
  \begin{aligned} \label{eq: R-N Derivative Q Expression}
    \left( \frac{d\mathbb{P}}{d\mathbb{Q}} \right)_t =
    \zeta_t &= \exp\left\{ {\sigma^{-2} \int_0^t  A_{u-} (dF_u - b\: ( dN_t^+ - dN_t^-)  )  - \frac{\sigma^{-2}}{2}\int_0^t \left( A_{u-} \right)^2 du } \right\} \\
    &\times
    \exp\left\{
    \int_0^t (1-\lambda_{u-}^+ ) du + \int_0^t \log(\lambda_{u-}^+ ) dN^+_u
    \right\} \\
    &\times
    \exp\left\{
    \int_0^t (1-\lambda_{u-}^- ) du + \int_0^t \log(\lambda_{u-}^- ) dN^-_u
    \right\}
    \;.
  \end{aligned}
\end{equation}
Using these last two expressions, we can re-write the filter in terms of $\mathbb{Q}$ expected values to obtain
\begin{align}
  \pi_t^j &= \frac{
  \E^{\mathbb{Q}} \left[ \1{\Theta_t=\theta_j} \zeta_t \mid \mcF_t \right]}
  {\E^{\mathbb{Q}} \left[ \zeta_t \mid \mcF_t \right]}
  \\ &=
  \frac{
  \E^{\mathbb{Q}} \left[ \1{\Theta_t=\theta_j} \zeta_t \mid \mcF_t \right]}
  {\sum_{i=1}^J \E^{\mathbb{Q}} \left[ \1{\Theta_t=\theta_i} \zeta_t \mid \mcF_t \right]}
  \\ &=
  \frac{\Lambda_t^j}{\sum_{i=1}^J \Lambda_t^i}
  \;.
\end{align}

Next, we will attempt to find an SDE for each $\Lambda_t^j$ term. This can be done by first defining the process $\delta_t^j = \1{\Theta_t = \theta_j}$, which satisfies the SDE
\begin{equation*}
  d\delta_t^j = \sum_{i=1}^J \delta_{t-}^i  C_{j,i} \: dt + d\mathcal{\tilde M}_t^j
  \;,
\end{equation*}
under the measure $\mathbb{Q}$, where $\mathcal{\tilde M}_t^j$ is a square-integrable, $\mcG_t$--adapted, $\mathbb{Q}$-martingale and $\bm C$ is the generator matrix for $\Theta_t$. All that's left to compute the dynamics of the $\Lambda_t^j$ is to figure out the dynamics of $\zeta_t \delta_t^j$ and then take the appropriate $\mathbb{Q}$ expected value while conditioning on $\mcF_t$. The process $\delta_t^j \zeta_t$ satisfies the SDE
\begin{align*}
  d\left( \zeta_{t} \delta_{t}^j \right) &=
  \left( \zeta_{t-} \delta_{t-}^j \right) \Bigg(
  \sigma^{-2}  A_{t-}^j \; ( dF_t - b (dN_t^+ - dN_t^-) )
  \\ &\hspace{1.5cm}+
  (\lambda_{t-}^+ -1)(dN_t^+ - dt) +
  (\lambda_{t-}^- -1)(dN_t^- - dt)
  \Bigg)
  \\
  &\hspace{1.5cm}
  + \sum_{i=1}^J \zeta_{t-} \delta_{t-}^i C_{j,i} \;dt + d\mathcal{M}_t^{j}
  \\
  &=
  \left( \zeta_{t-} \delta_{t-}^j \right) \Bigg(
  \sigma^{-2}  A_{t-}^j \; ( dF_t - b (dN_t^+ - dN_t^-) ) \\&\hspace{1.5cm}+
  (\lambda_{t-}^{+,j} -1)(dN_t^+ - dt) +
  (\lambda_{t-}^{-,j} -1)(dN_t^- - dt)
  \Bigg)
  \\
  &\hspace{1.5cm}
  + \sum_{i=1}^J \zeta_{t-} \delta_{t-}^i C_{j,i} \;dt + d\mathcal{M}_t^{j}
  \;,
\end{align*}
where $\mathcal{M}_t^{j}$ is another square-integrable, $\mcG_t$--adapted, $\mathbb{Q}$-martingale.

Now we re-write the expression for $\Lambda_t^j$ as the expected value of a stochastic integral

\begin{align}
  \Lambda_t^j = \E^{\mathbb Q} \Bigg[ &\delta_t^j \zeta_t \mid \mcF_t \Bigg]
  \\ = \E^{\mathbb Q} \Bigg[ &\left( \zeta_{0} \delta_{0}^j \right) +
  \int_0^t \left( \zeta_{u -} \delta_{u -}^j \right) \bigg(
  \sigma^{-2}  A_{u-}^j \; ( dF_u  - b (dN_u ^+ - dN_u ^-) ) \bigg) \label{eq: Proof Fubini Expression} \\
  +&\int_0^t \left( \zeta_{u -} \delta_{u -}^j \right) (\lambda_{u-}^{+,j} -1)(dN_u^+ - du)
  +\int_0^t \left( \zeta_{u -} \delta_{u -}^j \right) (\lambda_{u-}^{-,j} -1)(dN_u^- - du) \nonumber
  \\ +&\int_0^t  \left( \sum_{i=1}^J \zeta_{u-} \delta_{u-}^i C_{j,i} \;du + d\mathcal{M}_u^{j} \right)
  \Biggr\rvert \mcF_t \Bigg] \nonumber
  \;.
\end{align}
At this point we will need conditions guaranteeing that we can exchange the order of integration, in the expression above. First, by the definition of $\Lambda_t^j$,
\begin{equation}
	\mathbb{E}^{\mathbb{Q}}\left[
	\zeta_0 \delta_0^j
	\mid \mcF_t\right]
	= \pi_0^j
	\;,
\end{equation}
which allows us to replace the first term in \eqref{eq: Proof Fubini Expression}. We can use the fact that $F_t - F_0 - b(N_t^+ - N_t^-)=W_t^{\mathbb{Q}}$ is a $\mathbb{Q}$-Brownian motion to write
\begin{equation} \label{eq: Filter Theorem Formula 1}
	\mathbb{E}^{\mathbb{Q}}\left[
	\int_0^t \left( \zeta_{u -} \delta_{u -}^j \right)
  \sigma^{-2}  A_{u-}^j \; dW_u^{\mathbb Q}
	\mid \mcF_t\right]
	\;.
\end{equation}

Applying \cite[Chapter VI, Lemma 3.2]{Wong1985}, whose conditions are met with the bound \eqref{eq:NovikovCondition}, we can exchange the order of integration in the above term to get
\begin{equation}
	\int_0^t \Lambda_{u-}^j
  \sigma^{-2}  A_{u-}^j \; dW_u^{\mathbb Q}
  \;.
\end{equation}

For the third and fourth terms, let us note that if we let $\mathcal{U}^\pm_t = \left\{ u\in[0,t] : N_u^\pm > N_{u-}^\pm \right\}$, then by condition~\eqref{eq:NovikovCondition} we know $\mathcal{U}^\pm$ is almost surely finite and $\mcF_t$-measurable. Therefore
\begin{align}
	\mathbb{E}^{\mathbb{Q}}\left[
	\int_0^t \left( \zeta_{u -} \delta_{u -}^j \right) (\lambda_{u-}^{\pm,j} -1) dN_u^\pm
	\mid \mcF_t\right]
	&= \mathbb{E}^{\mathbb{Q}}\left[
	\sum_{u\in\mathcal{U}^\pm_t} \left( \zeta_{u -} \delta_{u -}^j \right) (\lambda_{u-}^{\pm,j} -1) ( N_u^\pm - N_{u-}^\pm )
	\mid \mcF_t\right]
	\\ &=
	\sum_{u\in\mathcal{U}^\pm_t} \Lambda_{u-}^j (\lambda_{u-}^{\pm,j} -1) ( N_u^\pm - N_{u-}^\pm )
  \\ &=
  \int_0^t \Lambda_{u-}^j (\lambda_{u-}^{\pm,j} -1) dN_t^{\pm}
	\;.
\end{align}
We can apply Fubini's theorem on the remaining Riemann integral because the integrands are square integrable to exchange the order of integration and get
\begin{align*}
	&\mathbb{E}^{\mathbb{Q}}\left[
	\int_0^t  \left( \sum_{i=1}^J \zeta_{u-} \delta_{u-}^i C_{j,i} -
	\left( \zeta_{u -} \delta_{u -}^j \right) (\lambda_{u-}^{+,j} + \lambda_{u-}^{-,j}-2)
	\right)\;du
	\mid \mcF_t\right]
	\\&\hspace{3.25cm}=
	\int_0^t  \left( \sum_{i=1}^J \Lambda_{u-}^i C_{j,i} -
	\Lambda_{u-}^j(\lambda_{u-}^{+,j} + \lambda_{u-}^{-,j}-2)
	\right)\;du
	\;.
\end{align*}

 By using the last steps to exchange the order of integration for each part, we can let the martingale ($\mathcal{M}_t^j$) portion vanish to obtain
\begin{align*}
  \Lambda_t^j = \;\pi_0^j \; + &\int_0^t \Lambda_{u-}^j \bigg(
  \sigma^{-2}  A_{u-}^j \; ( dF_u  - b (dN_u ^+ - dN_u ^-) ) \bigg) \\ +
  &\int_0^t \Lambda_{u-}^j (\lambda_{u-}^{+,j} -1)(dN_u^+ - du)
  + \int_0^t \Lambda_{u-}^j (\lambda_{u-}^{i,j} -1)(dN_u^- - du)
  + \sum_{i=1}^J \int_0^t \Lambda_{u-}^i C_{j,i} du
  \;,
\end{align*}
which is the desired result.
\end{proof}

\section{Proof of Theorem~\ref{th: Predictable Representation Theorem }}
\label{sec: Proof of Predictable Representation Theorem}
\begin{proof}

First assuming that $\sigma>0$, we shall prove the claims of Theorem~\ref{th: Predictable Representation Theorem } in order. First of all, it is clear in the definition of the process $\widehat W_t$ that it is a $\PP$-almost-surely continuous process satisfying $[\widehat W]_t = t$ because it is the sum of a $\mathbb{P}$-Brownian Motion and a process of finite variation. Moreover, by the definition of the process $F_t$, we can write $W_t$ as
\begin{equation}
    W_t = \sigma^{-1} \left( (F_t-F_0)
    - \int_0^t  A_u \; du
    - b(N_t^+ - N_t^-)
    \right)
    \;.
\end{equation}
Hence we can insert this last formula into the definition of $\widehat W_t$ to yield,
\begin{equation}
   \widehat W_t = \sigma^{-1}\left( (F_t-F_0)
    - \int_0^t   \hA_u  \; du
    - b(N_t^+ - N_t^-) \right)
    \;,
\end{equation}
which demonstrates that the process $\widehat W_t$ is $\mcF_t$--adapted. Next, we will show that $\widehat W_t$ is a $\PP$-martingale with respect to the filtration $\mcF_t$. By taking the conditional expectation of $\widehat W_{t+h}$ for $h\geq0$ and by using the properties of $W$, we get
\begin{align*}
  \mathbb E \left[ \widehat W_{t+h} \mid \mcF_t \right]
  &=
  \widehat W_t + \mathbb E \left[ \widehat W_{t+h} - \widehat W_t \mid \mcF_t \right]
  \\ &=
  \widehat W_t + \sigma^{-1} \mathbb E \left[ \int_t^{t+h}  \left( dF_u - \hA_u \, du - b \, (dN_u^+ - dN_u^-) \right) \, \mid \mcF_t \right]
  \\ &=
  \widehat W_t
  -  \sigma^{-1} \mathbb E \left[ \int_{t}^{t+h} \left( \hA_u -  A_u\right) \; du \mid \mcF_t \right]
  \\ &= \widehat W_t -
  \sigma^{-1} \mathbb E \left[ \int_{t}^{t+h}  \mathbb E \left[
   \hA_u -  A_u
   \mid \mcF_u \right] \; du \mid \mcF_t \right]
  \\ &= \widehat W_t
  \;,
\end{align*}
where in the above, the use of Fubini's theorem is allowed due to equation~\eqref{eq: Square Integrable A}. We have shown that $\widehat W$ is a $\mathbb{P}$-a.s. continuous martingale with quadratic variation equal to $t$. Therefore by the L\'evy characterization of Brownian Motion, the process $\widehat W$ is an $\mcF_t$--adapted $\mathbb{P}$-Brownian Motion.

Next we need to verify the claims made about the processes $\widehat M_t^\pm$. By the definitions of $\widehat M_t^\pm$ and of $M_t^\pm$,
\begin{equation} \label{eq: Theorem - Mhat N def}
  \widehat M_t^\pm = N_t^\pm - \int_0^t \widehat \lambda_u^\pm du
  \;,
\end{equation}
where the shorthand $\widehat \lambda_t^\pm$ is described after equation~\eqref{eq: F predictable dynamics}. Because the processes $\widehat \lambda_t^\pm$ are $\mcF_t$--adapted, we get that $\widehat M_t^\pm$ must also be $\mcF_t$--adapted processes. The processes $\widehat M_t^\pm$ are $\mcF_t$-martingales because for any $h>0$,
\begin{align}
  \mathbb{E} \left[ \widehat M_{t+h}^\pm \mid \mcF_t \right]
  &=
  \widehat M_t^\pm + \mathbb{E} \left[ M_{t+h}^\pm - M_{t}^\pm + \int_t^{t+h}
  ( \lambda_u^\pm - \widehat \lambda_u^\pm) \; du \mid \mcF_t \right]
  \\ &=
  \widehat M_t^\pm
  + \mathbb{E} \left[ \int_t^{t+h}
  \mathbb{E} \left[
  ( \lambda_u^\pm - \widehat \lambda_u^\pm)
  \mid \mcF_u \right]
  \;du
  \mid \mcF_t \right]
  \\ &= \widehat M_t^\pm
  \;.
\end{align}

By the definition of $\widehat M^\pm$ in equation~\eqref{eq: Theorem - Mhat N def}, $\widehat M^\pm$ is the sum of a process with an almost-surely finite number of jumps in the interval $[0,T]$ and a process of finite variation. From its definition, $\widehat W$ is the sum of a Brownian Motion and and a process of finite variation. The last two remarks imply that $[\widehat W, \widehat M^\pm]=0$. Lastly, because $dM_t^\pm = dN_t^\pm - \widehat \lambda_t^\pm dt$, and $d[N^+,N^-]_t =0$ almost surely, we get that $[\widehat M^+,\widehat M^-]_t = 0$ almost surely.

Because $N_t^\pm$ are a counting processes and $N_t^\pm - \int_0^t \hlambda_u \, du$ are $\mcF$--adapted $\mathbb{P}$ martingales, by Watanabe's characterization theorem, $N_t^\pm$ must be $\mcF$--adapted doubly stochastic Poisson processes with respective $\mathbb{P}$-intensities $\hlambda^\pm$, demonstrating claim (D).

If $\sigma=0$ and $A:=0$ the results concerning $\widehat{M}^\pm$ and $N^\pm$ still hold, and thus the statements $(B)$ and $(D)$ are true. By the same logic we also find that $[\widehat M^+,\widehat M^-]_t = 0$ almost surely.

\end{proof}

\section{Proofs related to the DPE}

\subsection{Proof of Proposition~\ref{prop: Candidate Solution}}\label{proof: Prop Candidate}
\begin{proof}
  Let us begin with the PDE~\eqref{eq: unsuped HJB} for $H(t,\bZ)$, where $\bZ=(F,\bN,X,Q,\blambda,\bLambda)$,
  \begin{equation}
    \begin{cases}
      0 = -\phi Q^2 + \sup_{\nu\in\mathds{R}} \left\{
      (\partial_t + \mathcal{ \bar L})H
      +\nu\,\partial_Q  H - \nu\,( F + \beta \left( Q - \mfN \right) + a \nu )\partial_X  H
    \right\} \\
    H(T,\bm Z) = X + Q( F + \beta \left( Q - \mfN \right)  - \alpha Q)
    \end{cases}
    \;.
  \end{equation}
  Because the term inside of the curly brackets is quadratic in $\nu$, we can complete the square and simplify the supremum expression. This yields
  \begin{equation}
     0 = -\phi Q^2 +
      (\partial_t + \mathcal{ \bar L})H
      + \frac{1}{4a}\frac{ \left( \partial_Q H - \left( F + \beta \left( Q - \mfN \right)  \right) \partial_X H  \right)^2}{\partial_X H}
    \;.
  \end{equation}
  Next, we can insert the ansatz
  \begin{equation}
    H(t,\bm Z)
    = X + Q \left( F + \beta \left( Q - \mfN \right) \right) + h(t,\bell(\bZ))
    \,
  \end{equation}
  where $\bell (\bZ) = (F,\bN,Q,\blambda,\bLambda)$, into the last PDE to yield another PDE in terms of $h$,
  \begin{equation}
    \begin{cases}
      0 = - \phi Q^2 + (\partial_t + \mathcal{\bar L})h
      + Q \left( \hA (t,F,\bN,\bLambda)
      + b ( \widehat \lambda^+(\blambda,\bLambda) -
            \widehat \lambda^-(\blambda,\bLambda) \right)
      + \frac{1}{4a} \left( \beta Q + \partial_Q h \right)^2
      \\
      h(T,\bell(\bZ)) = -\alpha Q^2
    \end{cases}
    \;.
  \end{equation}
  In the above equation, the term $\hA (t,F,\bN,\bLambda)
      + b ( \widehat \lambda^+(\blambda,\bLambda) -
            \widehat \lambda^-(\blambda,\bLambda) )$ appears in the PDE due to the mean drift of the process $F$, found in equation~\eqref{eq: F predictable dynamics}.
  If we assume that $h$ is quadratic in the variable $Q$ so that
  \begin{equation}
    h(t,\bell(\bZ)) =
    h_0(t,\bchi(\bZ))
  + Q\: h_1(t,\bchi(\bZ))
  + Q^2\: h_2(t)
  \;,
  \end{equation}
  where $\bchi(\bZ) = (F,\bN,\blambda,\bLambda)$, then the PDE further simplifies down to
  \begin{equation} \label{eq: Proof PDE Q Poly}
    \begin{cases}
      0 = \hspace{0.2cm} \left\{ (\partial_t + \mathcal{\bar L}) h_0 + \tfrac{1}{4a} h_1^2 \right\}
      \\
      \hspace{0.25cm} + Q \left\{ (\partial_t + \mathcal{\bar L})h_1
      + \left( \hA (t,F,\bN,\bLambda)
      + b ( \widehat \lambda^+(\blambda,\bLambda) -
            \widehat \lambda^-(\blambda,\bLambda) \right)
      +   \frac{1}{2a} (\beta + 2h_2) h_1
      \right\}
      \\
      \hspace{0.25cm} + Q^2 \left\{ \partial_t h_2 - \phi + \frac{1}{4a}\left(\beta + 2h_2\right)^2 \right\}

      \\
      h(T,\bell(\bZ)) = -\alpha Q^2
    \end{cases}
    \;.
  \end{equation}
  The above PDE must be satisfied for all values of $Q\in\mathds{R}$. Because $h_0$ $h_1$ and $h_2$ are independent of $Q$, each of the terms inside curly brackets in~\eqref{eq: Proof PDE Q Poly} must be equal to zero independently of $Q$. This yields the system of PDEs for $h_0$, $h_1$ and $h_2$
  \begin{align}
  &
    \begin{cases}
      0 = \partial_t h_0 + \frac{1}{4a} h_1^2
      \\ h_0(T,\bchi(\bZ)) = 0
    \end{cases}
  \\&
    \begin{cases} \label{eq: Prop h_1 PDE}
      (\partial_t + \mathcal{\bar L})h_1
      + \left( \hA (t,F,\bN,\bLambda)
      + b ( \widehat \lambda^+(\blambda,\bLambda) -
            \widehat \lambda^-(\blambda,\bLambda) \right)
      +   \frac{1}{2a} (\beta + 2 h_2) h_1
      \\ h_1(T,\bchi(\bZ)) = 0
    \end{cases}
  \\&
    \begin{cases}
      \partial_t h_2 - \phi + \frac{1}{4a}\left(\beta + 2 h_2\right)^2 = 0
      \\
      h_2(T) = -\alpha
    \end{cases}
  \end{align}
  which, due to their dependence on one another, can be solved in the order $h_2 \rightarrow h_1 \rightarrow h_0$. The ODE for $h_2$ is a standard Riccati-type ODE which admits the unique solution defined in the statement of the proposition. Next, the PDE for $h_1$ is linear and depends only on the solution for $h_2$. Therefore we can use the Feynman-Kac formula to write the solution PDE~\eqref{eq: Prop h_1 PDE} as
  \begin{equation}
    h_1(t,\bchi(\bZ))
    =
    \underset{t,\bchi(\bZ)}{\mathbb{E}}
    \left[
    \int_t^T
    \left(
    \hA_u
    + b \left( \widehat \lambda_u^+ - \widehat \lambda_u^- \right)
    \right)
    e^{\frac{1}{2a} \int_t^u (\beta + 2 h_2(\tau)) \; d\tau}
    \;du
    \right]
    \;,
  \end{equation}
  where the $\hA_u$ and the $\hlambda_u^{\pm}$ are defined in Section~\ref{sec: Reducing the Problem}. When plugging in the solution for $h_2$, we get the exact form presented in the statement of the proposition. Furthermore, condition~\eqref{eq: Square Integrable A} and the fact that the term $e^{\frac{1}{2a} \int_t^u (\beta + 2 h_2(\tau)) \; d\tau}$ is bounded for all $0\leq t\leq u \leq T$ allow us to use Fubini's theorem to yield
    \begin{equation}
    h_1(t,\bchi(\bZ))
    =
    \int_t^T
    \underset{t,\bchi(\bZ)}{\mathbb{E}}
    \left[
    \hA_u
    + b \left( \widehat \lambda_u^+ - \widehat \lambda_u^- \right)
    \right]
    e^{\frac{1}{2a} \int_t^u (\beta + 2 h_2(\tau)) \; d\tau}
    \;du
    \;,
  \end{equation}
  which in combination with the solution for $h_2$, gives us the form presented in the statement of the proposition.

  Lastly, the PDE for $h_0$ is also linear, and we can therefore use the Feynman-Kac formula once more for a representation of the solution. This representation gives us the expression for $h_0$ that is present in the statement of the proposition. We can also guarantee that the solution for $h_0$ provided in proposition~\ref{prop: Candidate Solution} is bounded because $\expe{}{\int_0^T (h_{1,u})^2 du}<\infty$. This bound is shown in equations~\ref{eq: Proof h1 bounded 1}--\ref{eq: Proof h1 bounded 2} in the proof of theorem~\ref{theorem: Verification Theorem}.

\end{proof}

\subsection{Proof of Theorem~\ref{theorem: Verification Theorem}}
\label{sec: Proof Verification Theorem}
\begin{proof}

    {\bf Showing the control $\nu^\star$ is admissible.}\\
  The candidate optimal control $\nu^\star$ is defined as
  \begin{equation}
    \nu^\star_t = \frac{1}{2a} \left( {Q_t^{\nu^\star}\left(\beta + 2 h_2(t) \right) + h_1(t,\bchi(\bZ_{t-})) } \right)
    \;,
  \end{equation}
  where $\bchi(\bZ) = (F,\bN,\blambda,\bLambda)$, and $\bZ = (F,\bN,X,Q,\blambda,\bLambda)$.
  It is clear from the definition above that the control is $\mcF$--adapted, because it is a continuous function of $\mcF$--adapted processes. To guarantee that the control $\nu_t^\star$ is admissible we must show that
  \begin{equation} \label{eq: Admissibility Criterion Proof}
    \expe{}{\int_0^T (\nu_u^\star)^2 du} < \infty
    \;.
  \end{equation}
  By expanding the expression for $(\nu_u^\star)^2$ and by using Young's inequality twice, we can write an upper bound for $(\nu_u^\star)^2$ as
  \begin{equation}
    (\nu_u^\star)^2 \leq \left(\frac{1}{2 a^2}\right) \left( (Q_u^{\nu^\star})^2 + (\beta + 2 h_2(u))^2 + (h_{1,u})^2  \right)
    \;,
  \end{equation}
  where $h_{1,u} =  h_1(u,\bchi(\bZ_{u-}))$. This last inequality shows that equation~\eqref{eq: Admissibility Criterion Proof} holds if each of $\expe{}{\int_0^T (\beta + 2 h_2(u))^2 du}$, $\expe{}{\int_0^T (Q_u^{\nu^\star})^2 du}$,
  and $\expe{}{\int_0^T (h_{1,u})^2 du}$ are bounded.

  Using the definition of $h_2$ in proposition~\ref{prop: Candidate Solution}, we can integrate the first term directly to obtain
  \begin{equation}
    \expe{}{\int_0^T (\beta + 2 h_2(u))^2 du}
    =
    a^2 \gamma^2 \left( T + \frac{2}{\gamma} \right) \left( \frac{1}{1-\zeta^\prime e^{2T\gamma}}
    - \frac{1}{1-\zeta^\prime} \right)
    \;,
  \end{equation}
  where $\zeta^\prime=\frac{\alpha - \frac{1}{2}\beta }{a \gamma}$. This last expression is bounded because $\alpha-\frac{1}{2}\beta \neq a \gamma$.

  Next, we can use the definition of $h_{1,u}$ provided in proposition~\ref{prop: Candidate Solution} to write
  \begin{equation} \label{eq: Proof h1 bounded 1}
    \expe{}{\int_0^T (h_{1,t})^2 dt} =
    \frac{1}{16 a^2} \expe{}{\int_0^T
    \left(
    \int_t^T
    \underset{t,\bchi_t}{\mathbb{E}}
    \left[
    \hA_u + b(\wlambda_u^{+} - \wlambda_u^{-})
    \right]
    \, \left(\frac{\zeta e^{\gamma\left( T-u \right)}  - e^{-\gamma\left( T-u \right)}}{\zeta e^{\gamma\left( T-t \right)} - e^{-\gamma\left( T-t \right)}}\right) du \right)^2 \, dt}
    \;,
  \end{equation}
  where $\bchi_t = \bchi(Z_t)$.
  Now if we notice that because $\gamma\geq0$, $\left(\frac{\zeta e^{\gamma\left( T-u \right)}  - e^{-\gamma\left( T-u \right)}}{\zeta e^{\gamma\left( T-t \right)} - e^{-\gamma\left( T-t \right)}}\right)^2 \leq 1 $ and that
  \begin{equation}
    \underset{t,\bchi_t}{\mathbb{E}}
    \left[
    \hA_u + b(\wlambda_u^{+} - \wlambda_u^{-})
    \right]
    =
    \underset{t,\bchi_t}{\mathbb{E}}
    \left[
     A_u + b(\lambda_u^{+} - \lambda_u^{-})
    \right]
    \;,
  \end{equation}
  then we can apply Jensen's inequality and Fubini's theorem, followed by Young's inequality to obtain
  \begin{align}
    \expe{}{\int_0^T (h_{1,t})^2 dt} &\leq
    \frac{1}{4 a^2} \int_0^T
    \int_t^T
    \mathbb{E}
    \left[
     A^2_u + b^2 \left( (\lambda_u^{+})^2 + (\lambda_u^{-})^2 \right)
    \right]
    \,  du  \, dt
    \\ &\leq
        \frac{T}{4 a^2} \int_0^T
    \mathbb{E}
    \left[
     A^2_u + b^2 \left( (\lambda_u^{+})^2 + (\lambda_u^{-})^2 \right)
    \right]
    \,  du
    \label{eq: Proof h1 bounded 2}
    \;.
  \end{align}
  By the condition of equation~\ref{eq: Square Integrable A}, this last term is bounded.

  By the definition of $Q_t^{\nu^\star}$ and of $\nu^\star$, we have that
  \begin{equation}
    dQ_t^{\nu^\star} = \frac{1}{2a}  \left( {Q_t^{\nu^\star}\left(\beta + 2 h_2(t) \right) + h_{1,t} } \right)\, dt
    \;, Q_0^{\nu^\star}=\mfN\;.
  \end{equation}
  The above SDE has the solution
  \begin{equation}
     Q_t^{\nu^\star} = \mfN + \frac{1}{2a} \int_0^t h_{1,u} \left(\frac{\zeta e^{\gamma\left( T-u \right)}  - e^{-\gamma\left( T-u \right)}}{\zeta e^{\gamma\left( T-t \right)} - e^{-\gamma\left( T-t \right)}}\right) \; du
     \;.
  \end{equation}
  By using Young's inequality and Jensen's inequality again, and by using the fact that ${\left(\frac{\zeta e^{\gamma\left( T-u \right)}  - e^{-\gamma\left( T-u \right)}}{\zeta e^{\gamma\left( T-t \right)} - e^{-\gamma\left( T-t \right)}}\right)^2 \leq 1 }$, then we can write
  \begin{align}
    (Q_t^{\nu^\star})^2 \leq \frac{1}{a} \left( \mfN^2 + \int_0^t (h_{1,u})^2 du \right)
    \;.
  \end{align}
  Now by taking the expectation and the integral of this last expression, we get
  \begin{align}
    \expe{}{\int_0^T (Q_u^{\nu^\star})^2 \; du}
    &\leq
    \frac{1}{a} \left( T\,\mfN^2 + \expe{}{\int_0^T \int_0^t (h_{1,u})^2 du\,dt} \right)
    \\ &\leq
    \frac{1}{a} \left( T\,\mfN^2 + T \expe{}{\int_0^T (h_{1,u})^2 du } \right)
    \;.
  \end{align}
  Because the term $\expe{}{\int_0^T (h_{1,u})^2 du }$ has already been shown to be bounded, we can conclude that $\expe{}{\int_0^T (Q_u^{\nu^\star})^2 \;du } < \infty$.

$\nu_t^{\star}$ is $\mcF_t$--adapted and satisfies $\expe{}{\int_0^T (\nu_u^\star)^2 du} < \infty$, therefore it is an admissible control. \\

  {\bf Showing $H\leq\widehat H$. }
  By applying It\^o's lemma to the function ${\widehat H(t,\bm Z) = X + Q \left( F + \beta(Q-\mfN) \right) + h(t,\bell(\bZ))}$ with an arbitrary control $\nu_t\in\mcA$ and the $\mcF_t$--predictable dynamics, we get
  \begin{align*}
    \widehat H_T
    = \widehat H(t,\bm Z)
    &+
    \int_t^T \left\{
    Q_u^\nu \left( \hA_u + b ( \wlambda_u^+ - \widehat \lambda_u^- ) \right)
    - a \nu_u^2 + (\beta + \partial_Q h_u ) \nu_u
    + \left( \partial_t + \mathcal{\bar L} \right) h_u
    \right\} du
    \\&+
    \int_t^T \eta_u^W \;d \widehat W_u
    + \int_t^T \eta_u^+ \; d \widehat M_u^+
    + \int_t^T \eta_u^- \; d \widehat M_u^-
    \;,
  \end{align*}
  where in the above, we use the notation $\widehat H_t = \widehat H(t,\bm Z_t)$, $h_t = h(t,\bZ_t)$ and we let $\eta_u^W$, $\eta_u^+$ and $\eta_u^-$ be the square-integrable $\mcF_t$--predictable processes obtained by the martingale representation theorem.

  By taking the conditional expected value of both sides, the Martingale portions vanish and we are left with
  \begin{align} \label{eq: Proof Verification Ito Expectation}
    \expe{t,\bZ}{\widehat H_T}
    &= \widehat H(t,\bZ)
    \\&+
    \expe{t,\bZ}{\int_t^T \left\{
    Q_u^\nu \left( \hA_u + b ( \wlambda_u^+ - \widehat \lambda_u^- ) \right)
    - a \nu_u^2 + (\beta + \partial_Q h_u ) \nu_u
    + \left( \partial_t + \mathcal{\bar L} \right) h_u
    \right\} du } \nonumber
    \;.
  \end{align}
  From the PDE~\eqref{eq: Reduced Ansatz 1 HJB PDE}, we get that for all $\nu\in\mathds{R}$,
  \begin{equation}
    0 \geq -\phi Q^2 +
    Q\,\left( \hA(t,F,\bN_t,\bLambda)
    + b (\widehat \lambda^+(\blambda,\bLambda ) -
    \widehat \lambda^-(\blambda,\bLambda )  \right)
    + (\beta Q + \partial_Q h)\nu - a \nu^2
    + \left( \partial_t + \mathcal{\bar L} \right) h
    \;.
  \end{equation}
  Therefore, by plugging in the above, as well as the boundary condition for $\widehat H_T$,
  \begin{align}
    \expe{t,\bZ}{\widehat H_T - \phi \int_t^T \left( Q_u^\nu \right)^2 du }
    =
    \expe{t,\bZ}{X_T^\nu + Q_t^\nu( F_T + \beta ( Q_T^\nu - \mfN ) -  \alpha Q_T^\nu) - \phi \int_t^T \left( Q_u^\nu \right)^2 du }
    \leq \widehat H(t,\bZ)
    \;.
  \end{align}
  By the definition of $H^\nu(t,\bZ)$ in equation~\eqref{eq: Value Function Definition}, and because the above holds for an arbitrary $\nu_t\in\mcA$, we obtain
  \begin{equation}
   H(t,\bZ) \leq \widehat H(t,\bZ)
   \;.
  \end{equation}

  {\bf Showing $H\geq H^{\nu^\star} \geq \widehat H$. }
  Next let us note that if we let $\nu^\star = \frac{\beta Q + \partial_Q h}{2a}$, then by equation~\eqref{eq: Reduced Ansatz 1 HJB PDE}, $\forall \varepsilon > 0$,
  \begin{equation}
    -\varepsilon < -\phi Q^2 +
    Q\,\left( \hA(t,F,\bN_t,\bLambda)
    + b (\widehat \lambda^+(\blambda,\bLambda ) -
    \widehat \lambda^-(\blambda,\bLambda )  \right)
    + (\beta Q + \partial_Q h)\nu^\star - a {\nu^\star}^2
    + \left( \partial_t + \mathcal{\bar L} \right) h
    \;.
  \end{equation}
  Using this last inequality with equation~\eqref{eq: Proof Verification Ito Expectation} and the definition of $H^\nu$ gives
  \begin{align*}
    H^{\nu^\star}(t,\bZ)
     &\geq \expe{t}{X_T^{\nu^\star} + Q_t^{\nu^\star}(S_T^{\nu^\star} -  \alpha Q_T^{\nu^\star}) - \phi \int_t^T Q_u^{\nu^\star} du } - \varepsilon
    \\&> \widehat H(t,\bZ)
    \;,
  \end{align*}
  Because $H\geq H^\nu$, $\forall\nu\in\mcA$ we get
  \begin{equation}
    H(t,\bZ)
    \geq
    H^{\nu^\star}(t,\bZ)
    \geq
    \widehat H(t,\bZ)
    \;.
  \end{equation}
  Therefore we obtain the desired result that
  \begin{equation}
    H = H^{\nu^\star} = \widehat H
  \end{equation}

\end{proof}

\section{Derivation of Forward-Backward Algorithm} \label{sec: Forward Backward Section}

This section provides further details on the forward-backward algorithm which allows computation of the smoother and two-slice marginal. The forward-backward algorithm presented here differs from what is usually found in the literature due to the fact that $Y$ may take continuous values, and that the process $Y$ is not conditionally independent in the usual way.  As Figure \ref{fig:HMM-graph} shows, there is also dependence between $Y$ even when conditioned on $Z$.

\subsection{Recursive Discrete Filter.}

To begin, define the sequence $\left\{ \alpha_n^{j,d} \right\}_{n=0}^{K-1}$ for each $j=1\dots J$ as $\alpha_n^{j,d} = \mathbb{P}\left( Z_n^d =\theta_j \mid \mathcal{Y}_{0:n}^d \right)$ -- the so-called \textit{forward-filter}. These filters satisfy a recursive relationship which we establish below. First note
\begin{equation}
	\alpha_0^{j,d} = \pi_0^j
	\;.
\end{equation}
Next, we can derive the recursion structure for this sequence by using applying Bayes' rule. Starting with the definition,
\begin{equation}
	 \alpha_n^{j,d} = \mathbb{P}\left( Z_n^d =\theta_j \mid \mathcal{Y}_{0:n}^d \right) =
	 \frac{ \mathbb{P}\left( Z_n^d =\theta_j , \mathcal{Y}_{0:n}^d \right)}{ \sum_{i=1}^J \mathbb{P}\left( Z_n^d =\theta_i , \mathcal{Y}_{0:n}^d \right)}
	 \;. \label{eqline: alpha 1}
\end{equation}
The numerator can be written recursively as
\begin{subequations}
\begin{align}
 \mathbb{P}\left( Z_n^d =\theta_j , \mathcal{Y}_{0:n}^d \right)
  &= \sum_{i=1}^J \mathbb{P}\left( \Znd{n}{j} , \Znd{n-1}{i} , \mathcal{Y}_{0:n}^d \right) \\
  &= \sum_{i=1}^J \mathbb{P}\left( \Znd{n}{j}  , \Ynd{n} \mid \Znd{n-1}{i}, \mathcal{Y}_{0:n-1}^d \right) \,
  \mathbb{P}\left( \Znd{n-1}{i} , \mathcal{Y}_{0:n-1}^d \right) \\
  &= \sum_{i=1}^J \mathbb{P}\left( Z_n^d =\theta_j  , \Ynd{n} \mid Z_{n-1}^d = \theta_i, \Ynd{n-1} \right) \,
  \P \left( \mathcal{Y}_{0:n-1}^d \right)
  \alpha_{n-1}^{i,d} \\
\begin{split}
  &= \sum_{i=1}^J \mathbb{P}\left( \Ynd{n} \mid \Znd{n-1}{j} , \Ynd{n-1}  \right)
  \\ &\phantom{\sum_{i=1}^J \mathbb{P}} \times \mathbb{P}\left( \Znd{n}{j} \mid \Znd{n-1}{i}, \Ynd{n-1} \right) \P \left( \mathcal{Y}_{0:n-1}^d \right) \,  \alpha_{n-1}^{i,d}
\end{split}
  \\ &=
  \P \left( \mathcal{Y}_{0:n}^d \right) \sum_{i=1}^J \bm P_{i,j} \, f_{\psi}(t_{n}\,y_{n}^d;t_{n-1},\theta_i,y_{n-1}^d\,)\, d\bm\mu(y_{n}^d) \,  \alpha_{n-1}^{i,d}
  \;.
\end{align}
\end{subequations}
 Therefore, by using the above result in equation~\eqref{eqline: alpha 1} and by canceling $\P \left( \mathcal{Y}_{0:n-1}^d \right)$ and $d\bm\mu(y_{n}^d)$ terms appearing in the numerator and in the denominator, we obtain
 \begin{align}
 	\alpha_n^{j,d} = \frac{\hat\alpha_n^{j,d}}{c_n^d}
 	\;,
 \end{align}
 where
 \begin{align}
 	\hat\alpha_n^{j,d} = \sum_{i=1}^J \bm P_{i,j} \, f_{\psi}(t_{n}\,y_{n}^d;t_{n-1},\theta_i,y_{n-1}^d\,)\,  \alpha_{n-1}^{i,d}\,,
 	\qquad \text{and} \qquad  c_n^d = \sum_{j=1}^J \hat\alpha_n^{j,d} \;.
 \end{align}

The normalization factor $c_n^d$ has the additional property that $c_n^d \, \bm\mu(dy_{n}^d) = \P \left( \Ynd{n} \mid \mathcal{Y}_{0:n-1}^d \right)$. This  can be seen by using the definition of the $\alpha_n^{j,d}$ and making use of the Markov property of $(Y,Z)$ as follows
\begin{subequations}
\begin{align}
	c_n^d \, \bm\mu(dy_{n}^d) &= \sum_{j=1}^J \sum_{i=1}^J \bm P_{i,j} \, f_{\psi}(t_{n},y_{n}^d;t_{n-1},\theta_j,y_{n-1}^d)\,
	d\bm\mu(y_{n}^d) \, \alpha_{n-1}^{i,d}
	\\&=
	\sum_{j=1}^J \sum_{i=1}^J \P\left( \Znd{n}{j} , \Ynd{n} \mid \Znd{n-1}{i} , \Ynd{n-1} \right)
	\mathbb{P}\left( Z_{n-1}^d =\theta_i \mid \mathcal{Y}_{0:n-1}^d \right)
	\\&=
	\sum_{j=1}^J \sum_{i=1}^J \P\left( \Znd{n}{j} , \Ynd{n} \mid \Znd{n-1}{i} , \mathcal{Y}_{0:n-1}^d \right)
	\mathbb{P}\left( Z_{n-1}^d =\theta_i \mid \mathcal{Y}_{0:n}^d \right)
	\\ &=
	\sum_{j=1}^J \sum_{i=1}^J \P\left( \Znd{n}{j} , \Znd{n-1}{i} , \Ynd{n} \mid
	\mathcal{Y}_{0:n-1}^d \right)
	\\&=
	\P\left( \Ynd{n} \mid \mathcal{Y}^d_{0:n-1} \right)
	\;.
\end{align}
\end{subequations}

\subsection{Recursive Backward Discrete Filter.}

Here, we derive the recursion for the \textit{backward-filter} $\{ \beta_n^{j,d} \}_{n=0}^{K-1}$ for each $j=1\dots J$, defined as
\begin{equation}
	\beta_n^{j,d} = \frac{ \P\left( \Y_{n+1:K}^d \mid \Znd{n}{j} , \Ynd{n} \right) }{ \P\left( \Y_{n+1:K}^d \mid \Y_{0:n}^d \right) }
	\;.
\end{equation}

Just as with the forward-filter, the backward-filter can be obtained recursively. First note that
\begin{equation}
	\beta_n^{j,d} = \frac{ \P\left( \Y_{n+1:K}^d \mid \Znd{n}{j} ,\Ynd{n} \right) }{ \P\left( \Y_{n+1:K}^d \mid \Y_{0:n}^d \right) }
=
	\frac{ \P\left( \Y_{n+1:K}^d \mid \Znd{n}{j} ,\Ynd{n} \right) }{
	\sum_{i=1}^J \P\left( \Y_{n+1:K}^d \mid \Znd{n}{i} ,\Ynd{n} \right) \alpha_{n}^i
	} \label{eqline: beta 1}
	\;,
\end{equation}
which can be computed at time $n=K-1$ as
\begin{equation}
	\beta_n^{j,d} =
	\frac{ f_{\psi}(t_{K}, y^d_{K} ; t_{K-1}, \theta_j, y^d_{K-1}) }{
	\sum_{i=1}^J f_{\psi}(t_{K}, y^d_{K} ; t_{K-1}, \theta_i, y^d_{K-1}) \alpha_{K-1}^i }
	\;.
\end{equation}
Continuing with expression for the numerator in equation~\eqref{eqline: beta 1}, we find that
\begin{subequations}
\begin{align}
& \P\left( \Y_{n+1:K}^d \mid \Znd{n}{j} ,\Ynd{n} \right)
	\\
	&\quad=
	\sum_{i=1}^J \P\left( \Znd{n+1}{i}, \Y_{n+1:K}^d \mid \Znd{n}{j} ,\Ynd{n} \right)
	\\
\begin{split}
&\quad=
	\sum_{i=1}^J \P\left( \mathcal{Y}^d_{n+2:K} \mid \Znd{n+1}{i} ,
	\Znd{n}{j} , \Y_{n:n+1}^d \right)
	\\&\qquad\qquad\times\P\left( \Znd{n+1}{i} , \Ynd{n+1} \mid Z_n^d = \theta_j, \Ynd{n} \right)
\end{split}
	\\
\begin{split}
&\quad=
	\sum_{i=1}^J \P\left( \mathcal{Y}^d_{n+2:K} \mid \Znd{n+1}{i} , \Ynd{n+1}  \right)
	\\&\qquad\qquad\times \bm P_{j,i} \, f_{\psi}(t_{n+1}, y^d_{n+1} ; t_{n}, \theta_j, y^d_{n}) \,d\bm\mu(y_{n+1}^d)
\end{split}
	\\ &\quad=
	\P\left( \Y_{n+2:K}^d \mid \Y_{0:n+1}^d \right)
	f_{\psi}(t_{n+1}, y^d_{n+1} ; t_{n}, \theta_j, y^d_{n}) \,d\bm\mu(y_{n+1}^d)
	\sum_{i=1}^J \beta_{n+1}^{i,d} \bm P_{j,i}
	\;.
\end{align}
\end{subequations}
Plugging this last result back into equation~\eqref{eqline: beta 1}, and canceling $d\bm\mu(y_{n+2}^d)$ and $\P\left( \Y_{n+3:K}^d \mid \Y_{0:n+2}^d \right)$ terms, we obtain
\begin{align}
	\beta_n^{j,d} = \frac{\hat\beta_n^{j,d}}{\sum_{i=1}^J \hat\beta_n^{i,d} \alpha_{n}^{i,d} }
	\;,
\end{align}
where
\begin{align}
	\hat\beta_n^{j,d} = f_{\psi}(t_{n+1}, y^d_{n+1} ; t_{n}, \theta_j, y^d_{n}) \,
	\sum_{i=1}^J \beta_{n+1}^{i,d} \bm P_{j,i}
	\;.
\end{align}

Furthermore, because $(Y,Z)$ is a markov process, by the Markov property
\begin{equation}
	\beta_n^{j,k} =
	P\left( \Y_{n+1:K}^d \mid \Znd{n}{j} ,\mathcal{Y}_{0:n} \right)
	=
	P\left( \Y_{n+1:K}^d \mid \Znd{n}{j} ,\Ynd{n} \right)
	\;,
\end{equation}
a fact which will be used a number of times in the next part.

\subsection{Expressions for the Discrete Smoother}

The main objective of this section is to compute the smoother and two-slice marginal, $\{ \gamma_n^{j,d}\}_{n=0}^{K-1}$ and $\{ \xi_n^{i,j,d}\}_{n=0}^{K-2}$. For convenience, we repeat their definition here
\begin{align*}
	\gamma_n^{j,d} = \P\left( \Znd{n}{j} \mid \mathcal{Y}^d_{0:K} \right),
	\qquad \text{and}\qquad
	\xi_n^{i,j,d} = \P\left( \Znd{n}{i} , \Znd{n+1}{j} \mid \mathcal{Y}^d_{0:K} \right)
\end{align*}
for all allowed values of $n$, and for each $i,j=1\dots J$.

To this end, note that
\begin{subequations}
\begin{align}
	\gamma_n^{j,d}= \P\left( \Znd{n}{j} \mid \mathcal{Y}^d_{0:K} \right)
	&= \frac{ \P\left( \Znd{n}{j} , \mathcal{Y}^d_{0:K} \right) }
	{ \P\left( \mathcal{Y}^d_{0:K} \right)}
	\\ &=
	\frac{ \P\left( \mathcal{Y}^d_{n+1:K} \mid \Znd{n}{j} , \mathcal{Y}^d_{0:n} \right) \,
	\P\left( \Znd{n}{j} \mid \mathcal{Y}^d_{0:n} \right) }
	{ \P\left( \mathcal{Y}^d_{n+1:K} \mid \mathcal{Y}^d_{0:n} \right)
	}
	\\ &=
	\alpha_n^{j,d}\;\beta_n^{j,d}
	\;.
\end{align}%
\end{subequations}%
Next,
\begin{subequations}
\begin{align}
	\xi_n^{i,j} &=
	\P\left( \Znd{n}{i} , \Znd{n+1}{j} \mid \mathcal{Y}^d_{0:K} \right)
	\\ &=
	\frac{\P\left( \Znd{n+1}{j} , \mathcal{Y}^d_{n+1:K} \mid
	\Znd{n}{i}, \mathcal{Y}^d_{0:n} \right)
	\P\left( \Znd{n}{i} \mid \mathcal{Y}^d_{0:n} \right)}
	{\P\left( \mathcal{Y}^d_{n+1:K} \mid \mathcal{Y}^d_{0:n} \right)}
	\\ &=
	\alpha_n^{i,d} \left(
	\frac{\P\left( \mathcal{Y}^d_{n+2:K} \mid \Znd{n+1}{j} , \mathcal{Y}^d_{0:n+1} \right)
	\P\left( \Ynd{n+1} , \Znd{n+1}{j} \mid \Znd{n}{i} , \mathcal{Y}^d_{0:n} \right) }
	{\P\left( \mathcal{Y}^d_{n+2:K} \mid \mathcal{Y}^d_{0:n+1} \right)
	\P\left( \Ynd{n+1} \mid \mathcal{Y}^d_{0:n} \right)}
	\right)
	\\ &=
	\frac{\alpha_n^{i,d}\;\beta_{n+1}^{j,d}}{c_{n+1} \;d\bm\mu(y_{n+1}^d)}\;
	\P\left( \Ynd{n+1} , \Znd{n+1}{j} \mid \Znd{n}{i} , \mathcal{Y}^d_{0:n} \right)
	\\ &=
	\frac{\alpha_n^{i,d}\;\beta_{n+1}^{j,d}}{c_{n+1}^d} \; \bm P_{i,j} \;f_\psi(t_{n+1}, y_{n+1}^d ; t_{n}, \theta_i, y_{n}^d)
	\;.
\end{align}%
\end{subequations}%
The relationships between $(\gamma,\xi)$, and $(\alpha,\beta)$ are applied when performing the E-step in the EM algorithm described in Section~\ref{sec: EM Algorithm}.

The natural ordering of computation proceeds by first computing $\{ \alpha_n^{j,d} \}_{n=0}^{K-1}$, $\{ c_n^{d} \}_{n=0}^{K-1}$ and $\{ \beta_n^{j,d} \}_{n=0}^{K-1}$ in order, and then using the results to compute $\gamma$ and $\xi$.

\section{Calibration to INTC stock returns} \label{sec: Calibration Results}

\footnotesize

This section contains the results for the truncated pure-jump model described in Section~\ref{sec: Censored Pure Jump}, calibrated to per-second prices on INTC stocks. The calibrated parameters are displayed below for the models with $1$ to $6$ latent states. As mentioned in Section~\ref{sec: INTC Fit Discussion}, we use the BIC and ICL criterion to determine the `optimal' number of latent states. The BIC is defined as
\begin{equation}
  BIC = \log L^\star - \frac{\nu_M}{2} \,\log{(K \times D)}
  \;,
\end{equation}
where $\log L^\star$ is the value of the maximized log-likelihood for a given model,  $\nu_J$ is the number of parameters present in the model, and  recall that $D$ represents the number of observation days, and $K$ the number of observations within a day (assumed equal across days).

As discussed in Section~\ref{sec: EM Algorithm}, the log-likelihood cannot be computed directly. Instead, we use the Forward-Backward algorithm in Section~\ref{sec: Forward Backward Section} to compute it. The log-likelihood for a given model, using the notation of Section~\ref{sec: EM Algorithm}, can be computed as
\begin{equation}
  \log L = \sum_{d=1}^D \sum_{k=0}^{K-1} \log c_k^d
  \;.
\end{equation}
The approximation to the ICL of~\cite{biernacki2000assessing} for our model can be computed directly as
\begin{equation}
  ICL = \sum_{d=1}^D \sum_{k=0}^{K-1} \log f_{\psi^\star}( t_{k+1} , y_{k+1}^d ; t_{k} , y_{k}^d , \widehat Z_{k}^d  )
  - \frac{\nu_M}{2} \log{(K \times D)}
  \;,
\end{equation}
where $\nu_M$ is again the number of parameters present in the model. $\psi^\star$ are the parameters appearing in the transition density function $f_{\psi}$ which maximize the model's log-likelihood. $\widehat Z^d$ is the most likely path of $Z^d$ conditional on $\Y_{0:K}^d$, as computed by the Viterbi algorithm, using the parameters which maximize the model's likelihood.

The tables below record the calibrated parameters for the mean-reverting pure-jump model presented in Section~\ref{sec: Censored Pure Jump}. Each table contains the calibrated parameters using the EM algorithm for the number of possible states for the latent process ranging from 1 to 6. Each of the rows in the tables below are ordered by the base noise level $\mu_i$. Along with the parameters $\mu_i$, $\kappa_i$, and $\theta_i$, we also include the initial probability of the latent random variable starting in each of the given states ($\pi_0^i$) as well as the generator matrix for the latent process, which is computed as $\bm C = \log \bm P$.

\begin{table}[H]
\footnotesize
  \centering
  \caption{Mean-Reverting Pure-Jump model calibration results for INTC and $J=1$.}
    \begin{tabular}{rrr}
    \toprule
    \toprule
    \multicolumn{3}{c}{One Latent State} \\
    \midrule
    \multicolumn{1}{c}{$\mu_i$} & \multicolumn{1}{c}{$\kappa_i$} & \multicolumn{1}{c}{$\theta_i$} \\
    0.0334 & 0.0748 & 0.0200 \\
    \bottomrule
    \bottomrule
    \end{tabular}%
  \label{tab:INTC-1State}%
\end{table}%

\begin{table}[H]
\footnotesize
  \centering
  \caption{Mean-Reverting Pure-Jump model calibration results for INTC and $J=2$.}
    \begin{tabular}{rrrrrrr}
    \toprule
    \toprule
    \multicolumn{7}{c}{2 Latent States} \\
          &       &       &       &       & \multicolumn{2}{c}{$C_{i,j}$} \\
\cmidrule{6-7}          &       &       &       &       & \multicolumn{2}{c}{State $j$} \\
    State $i$ & \multicolumn{1}{c}{$\pi_0^i$} & \multicolumn{1}{c}{$\mu_i$} & \multicolumn{1}{c}{$\kappa_i$} & \multicolumn{1}{c}{$\theta_i$} & \multicolumn{1}{c}{1} & \multicolumn{1}{c}{2} \\
    \midrule
    1     & 0.7969 & 0.0899 & 0.0897 & 0.0200 & -0.00792 & 0.00792 \\
    2     & 0.2031 & 0.0183 & 0.0100 & 0.0201 & 0.00245 & -0.00245 \\
    \bottomrule
    \bottomrule
    \end{tabular}%
  \label{tab:INTC-2State}%
\end{table}%

\begin{table}[H]
\footnotesize
  \centering
  \caption{Mean-Reverting Pure-Jump model calibration results for INTC and $J=3$.}
  \begin{tabular}{rrrrrrrr}
    \toprule
    \toprule
    \multicolumn{8}{c}{3 Latent States} \\
        &       &       &       &       & \multicolumn{3}{c}{$C_{i,j}$} \\
    \cmidrule{6-8}      &       &       &       &       & \multicolumn{3}{c}{State $j$} \\
    State $i$ & \multicolumn{1}{c}{$\pi_0^i$} & \multicolumn{1}{c}{$\mu_i$} & \multicolumn{1}{c}{$\kappa_i$} & \multicolumn{1}{c}{$\theta_i$} & \multicolumn{1}{c}{1} & \multicolumn{1}{c}{2} & \multicolumn{1}{c}{3} \\
    \midrule
    1     & 0.4483 & 0.1727 & 0.0950 & 0.0100 & -0.0147 & 0.0147 & 0.0000 \\
    2     & 0.4510 & 0.0447 & 0.0288 & 0.0000 & 0.0020 & -0.0037 & 0.0017 \\
    3     & 0.1006 & 0.0139 & 0.0117 & 0.0200 & 0.0000 & 0.0012 & -0.0012 \\
    \bottomrule
    \bottomrule
  \end{tabular}%
  \label{tab:INTC-3State}%
\end{table}%

\begin{table}[H]
\footnotesize
  \centering
  \caption{Mean-Reverting Pure-Jump model calibration results for INTC and $J=4$.}
  \begin{tabular}{rrrrrrrrr}
  \toprule
  \toprule
  \multicolumn{9}{c}{4 Latent States} \\
        &       &       &       &       & \multicolumn{4}{c}{$C_{i,j}$} \\
  \cmidrule{6-9}      &       &       &       &       & \multicolumn{4}{c}{State $j$} \\
  State $i$ & \multicolumn{1}{c}{$\pi_0^i$} & \multicolumn{1}{c}{$\mu_i$} & \multicolumn{1}{c}{$\kappa_i$} & \multicolumn{1}{c}{$\theta_i$} & \multicolumn{1}{c}{1} & \multicolumn{1}{c}{2} & \multicolumn{1}{c}{3} & \multicolumn{1}{c}{4} \\
  \midrule
  1     & 0.3431 & 0.2457 & 0.1348 & -0.0213 & -0.0226 & 0.0225 & 0.0001 & 0.0000 \\
  2     & 0.2875 & 0.0658 & 0.0292 & 0.0100 & 0.0031 & -0.0057 & 0.0026 & 0.0000 \\
  3     & 0.3368 & 0.0305 & 0.0176 & 0.0000 & 0.0000 & 0.0009 & -0.0023 & 0.0014 \\
  4     & 0.0326 & 0.0100 & 0.0079 & 0.0000 & 0.0000 & 0.0000 & 0.0015 & -0.0015 \\
  \bottomrule
  \bottomrule
  \end{tabular}
  \label{tab:INTC-4State}%
\end{table}%

\begin{table}[H]
\footnotesize
  \centering
  \caption{Mean-Reverting Pure-Jump model calibration results for INTC and $J=5$.}
  \begin{tabular}{rrrrrrrrrr}
  \toprule
  \toprule
  \multicolumn{10}{c}{5 Latent States} \\
        &       &       &       &       & \multicolumn{5}{c}{$C_{i,j}$} \\
  \cmidrule{6-10}      &       &       &       &       & \multicolumn{5}{c}{State $j$} \\
  State $i$ & \multicolumn{1}{c}{$\pi_0^i$} & \multicolumn{1}{c}{$\mu_i$} & \multicolumn{1}{c}{$\kappa_i$} & \multicolumn{1}{c}{$\theta_i$} & \multicolumn{1}{c}{1} & \multicolumn{1}{c}{2} & \multicolumn{1}{c}{3} & \multicolumn{1}{c}{4} & \multicolumn{1}{c}{5} \\
  \midrule
  1     & 0.1145 & 0.2255 & 0.1805 & -0.0201 & -0.0040 & 0.0000 & 0.0040 & 0.0000 & 0.0000 \\
  2     & 0.5979 & 0.1373 & 0.0463 & -0.0095 & 0.0000 & -0.0835 & 0.0786 & 0.0049 & 0.0000 \\
  3     & 0.0000 & 0.0376 & 0.0188 & 0.0301 & 0.0003 & 0.0325 & -0.0328 & 0.0000 & 0.0000 \\
  4     & 0.2663 & 0.0269 & 0.0149 & 0.0222 & 0.0000 & 0.0006 & 0.0000 & -0.0020 & 0.0014 \\
  5     & 0.0213 & 0.0084 & 0.0065 & -0.0004 & 0.0000 & 0.0000 & 0.0000 & 0.0019 & -0.0019 \\
  \bottomrule
  \bottomrule
  \end{tabular}
  \label{tab:INTC-5State}%
\end{table}%

\begin{table}[H]
\footnotesize
  \centering
  \caption{Mean-Reverting Pure-Jump model calibration results for INTC and $J=6$.}
  \begin{tabular}{rrrrrrrrrrr}
  \toprule
  \toprule
  \multicolumn{11}{c}{6 Latent States} \\
        &       &       &       &       & \multicolumn{6}{c}{$C_{i,j}$} \\
  \cmidrule{6-11}      &       &       &       &       & \multicolumn{6}{c}{State $j$} \\
  State $i$ & \multicolumn{1}{c}{$\pi_0^i$} & \multicolumn{1}{c}{$\mu_i$} & \multicolumn{1}{c}{$\kappa_i$} & \multicolumn{1}{c}{$\theta_i$} & \multicolumn{1}{c}{1} & \multicolumn{1}{c}{2} & \multicolumn{1}{c}{3} & \multicolumn{1}{c}{4} & \multicolumn{1}{c}{5} & \multicolumn{1}{c}{6} \\
  \midrule
  1     & 0.5678 & 0.3206 & 0.0563 & 2.7090 & -0.4276 & 0.0003 & 0.2485 & 0.1405 & 0.0331 & 0.0051 \\
  2     & 0.0906 & 0.3080 & 0.2608 & -0.0800 & 0.0000 & -0.0103 & 0.0103 & 0.0000 & 0.0000 & 0.0000 \\
  3     & 0.0054 & 0.1078 & 0.0359 & -0.0137 & 0.0000 & 0.0012 & -0.0144 & 0.0131 & 0.0001 & 0.0000 \\
  4     & 0.1291 & 0.0405 & 0.0176 & -0.0021 & 0.0000 & 0.0000 & 0.0046 & -0.0055 & 0.0009 & 0.0000 \\
  5     & 0.1915 & 0.0260 & 0.0143 & 0.0238 & 0.0000 & 0.0000 & 0.0003 & 0.0000 & -0.0025 & 0.0022 \\
  6     & 0.0157 & 0.0074 & 0.0054 & 0.0161 & 0.0000 & 0.0000 & 0.0000 & 0.0000 & 0.0029 & -0.0030 \\
  \bottomrule
  \bottomrule
  \end{tabular}
  \label{tab:INTC-6State}%
\end{table}%

\clearpage

\bibliographystyle{chicago}
\bibliography{LatentAlphaModels_V11Ref}

\end{document}